\documentclass[twoside,leqno,twocolumn]{article}
\usepackage{ltexpprt}
\usepackage{amsmath}
\usepackage{array}
\usepackage{longtable}
\usepackage{graphicx}
\usepackage{epstopdf}
\usepackage{float}
\usepackage{multirow}
\usepackage{balance}
\usepackage{bm}
\usepackage{color}
\usepackage[dvipsnames]{xcolor}
\usepackage{algorithm}
\usepackage{algorithmicx}
\usepackage{algpseudocode}
\usepackage{makecell}
\usepackage{xspace}
\usepackage{amsfonts}
\usepackage{amssymb}
\usepackage{url}
\usepackage{caption}
\usepackage[font=scriptsize]{subcaption}
\usepackage[flushleft]{threeparttable}
\usepackage{upgreek}
\usepackage{bm}
\usepackage{hyperref}
\usepackage{cleveref}
\usepackage{enumerate}
\usepackage{enumitem}
\usepackage{multirow}
\usepackage{mathtools}
\usepackage{paralist}

\newcommand{\eaglemine}{{EagleMine}{ }}
\newcommand{\Eaglemine}{{EagleMine}}
\newcommand{\pltname}{histogram\xspace}

\newcommand{\mcls}{micro-clusters\xspace}

\newcommand{\hide}[1]{}

\newcommand{\bit}{\begin{compactitem}}
\newcommand{\eit}{\end{compactitem}}
\newcommand{\ben}{\begin{compactenum}}
\newcommand{\een}{\end{compactenum}}

\newtheorem{trait}{Trait}

\newtheorem{problem}{Problem}
\newtheorem{definition}{Definition}

\newcommand{\codeurl}{https://goo.gl/dJJp3n}

\renewcommand\labelenumi{\roman{enumi}}
\renewcommand\theenumi\labelenumi

\newcommand*\rotv{\rotatebox{90}}

\newcolumntype{L}{>{\centering\arraybackslash}m{3.2cm}}

\newcommand{\norm}[1]{\left\lVert#1\right\rVert}

\definecolor{light-gray}{gray}{0.95}
\definecolor{dark-gray}{rgb}{0.66, 0.66, 0.66}

\DeclareMathOperator*{\argmin}{arg\,min}

\begin{document}

\setlength{\abovedisplayskip}{0.1in}
\setlength{\belowdisplayskip}{0.1in}



\title{{EagleMine}: Vision-Guided Mining in Large Graphs}

\author{
Wenjie Feng$^{\dag}$, Shenghua Liu$^{\dag}$, Christos Faloutsos$^{\ddag}$, 
Bryan Hooi$^{\ddag}$, Huawei Shen$^{\dag}$, Xueqi Cheng$^{\dag}$ \\
$^{\dag}$CAS Key Laboratory of Network Data Science \& Technology,\\
Institute of Computing Technology, Chinese Academy of Sciences\\
$^{\ddag}$School of Computer Science, Carnegie Mellon University\\
\small{fengwenjie@software.ict.ac.cn, liu.shengh@gmail.com, christos@cs.cmu.edu} \\ 
\small{bhooi@andrew.cmu.edu, \{shenhuawei, cxq\}@ict.ac.cn}
\vspace{-0.15in}
}
\date{}

\maketitle






\algrenewcommand{\algorithmicrequire}{\textbf{Input:}}
\algrenewcommand{\algorithmicensure}{\textbf{Output:}}


\begin{abstract}
\label{sec:abs}
Given a graph with millions of nodes, what patterns exist in the distributions 
of node characteristics, and how can we detect them and separate 
anomalous nodes in a way similar to human vision?
In this paper, we propose a vision-guided algorithm,
\Eaglemine, to summarize micro-cluster patterns in two-dimensional histogram 
plots constructed from node features in a large graph.
\eaglemine utilizes 
a water-level tree to capture cluster structures according to 
vision-based intuition at multi-resolutions.
\eaglemine traverses the water-level tree from the root and adopts
statistical hypothesis tests to determine the optimal clusters that
should be fitted along the path, and summarizes each cluster with a truncated Gaussian distribution.
Experiments on real data show that our method can find truncated 
and overlapped elliptical clusters, even when some baseline methods
split one visual cluster into pieces with Gaussian spheres. 
To identify potentially anomalous micro-clusters, 
\eaglemine also a designates score to measure the suspiciousness of outlier 
groups (i.e. node clusters) and outlier nodes, detecting bots and anomalous 
users with high accuracy in the real Microblog data.

\end{abstract}

\section{Introduction}
\label{sec:intro}
Given real-world graphs with millions of nodes and connections,
how can we separate the nodes into groups and summarize 
the graph in an intuitive and persuasive way
being consistent with human vision? 
In particular, the graph can represent friendships in Facebook, 
rating or retweeting behaviors between users 
and objects in Amazon and Twitter.
A series of topics in human-computer interaction (HCI) 
study graph visualization by projecting 
the nodes into a two-dimensional plot to facilitate manual 
patterns recognition\cite{Kang2014NetRayVA,koutra2015perseus}.
By incorporating with visualization technologies and HCI tools, 
graph mining and clustering become more intuitive and interpretable. 
In addition, visualization-based methods are often more scalable 
compared with other approaches focusing on matrix or tensor 
analysis for big graphs as they operate on histograms 
rather than the raw data. 
However, there are hundreds of ways to visualize a static graph 
in a plot by selecting different characteristics, e.g.
degree, triangles, eigen/singular vectors, PageRank, 
and even more plots for snapshot graphs of a dynamic system. 
It is labor-consuming to manually recognize patterns 
from those plots.
How can we automatically find and summarize 
micro-clusters, interesting patterns, and anomalies in these plots 
for mining large graphs?

In this paper, we propose \Eaglemine, a novel tree-based
mining approach to find and summarize all the node clusters 
in a general \pltname plot of a graph.
Inspired by the human vision, 
\eaglemine discovers a hierarchy of clusters and
describes the node distribution in each cluster with a
model vocabulary. 
\eaglemine also summarizes interesting and anomalous node 
clusters besides the majority cluster in a graph.
In the experiments, we test \eaglemine on real data, and
publicly-available graphs.

In summary, the proposed \eaglemine has the
following desirable properties:

\begin{compactitem}
    \item \textit{Automated summarization:} 
	\eaglemine automatically summarizes a given histogram with
	a vocabulary of distributions, which separates 
	graph nodes into groups and outliers as human vision does.
    \item \textit{Effectiveness:}  \eaglemine 
	detects interpretable clusters, and outperforms 
	the baselines and even the manually tuned competitors. 
	(see Figure~\ref{fig:wbdbscan}-\ref{fig:wbeaglemine}), 
	obtaining better summarization performance.
	 For anomaly detection on real data, \eaglemine also achieves 
	 higher accuracy, finding micro-clusters of suspicious users 
	 and bots (see Figure~\ref{fig:sina_outdhub_cases}).
    \item \textit{Scalability:}
	\eaglemine can be easily extended and runs
	in linear-time in the number of graph nodes given 
	the visualization plot.
\end{compactitem}

\textbf{Reproducibility:} Our code is open-sourced 
at \url{\codeurl}, and most of the data we use 
is publicly available online.

\begin{figure*}[t]
	\centering
	\begin{subfigure}[t]{0.3\linewidth} 
		\centering
		\includegraphics[height=1.2in]{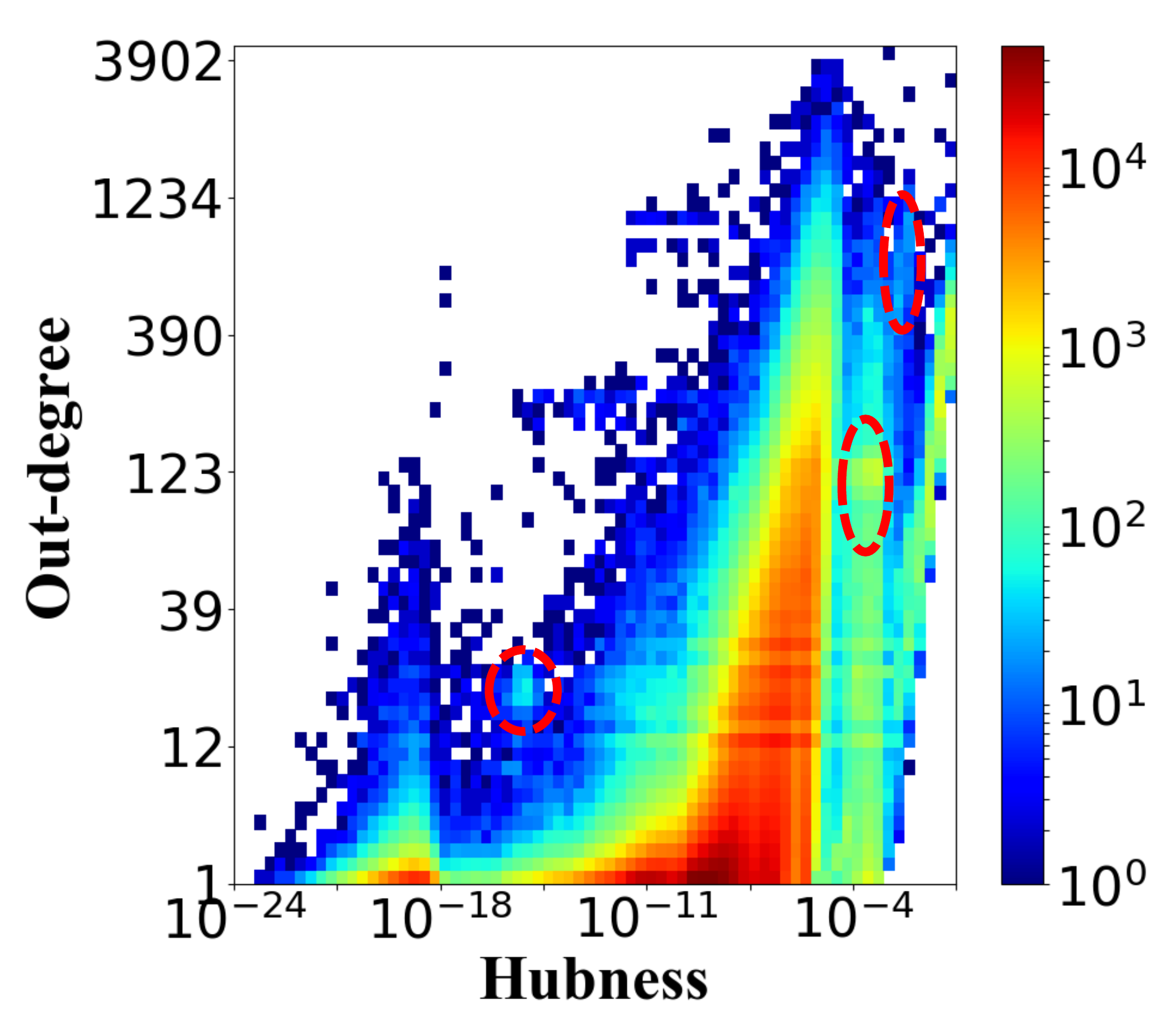}
		\vspace{-0.05in}
		\caption{Out-degree vs. Hubness}
		\label{fig:sina_outdhub}
	\end{subfigure} ~
	\begin{subfigure}[t]{0.25\linewidth} 
		\centering
		\includegraphics[height=1.2in]{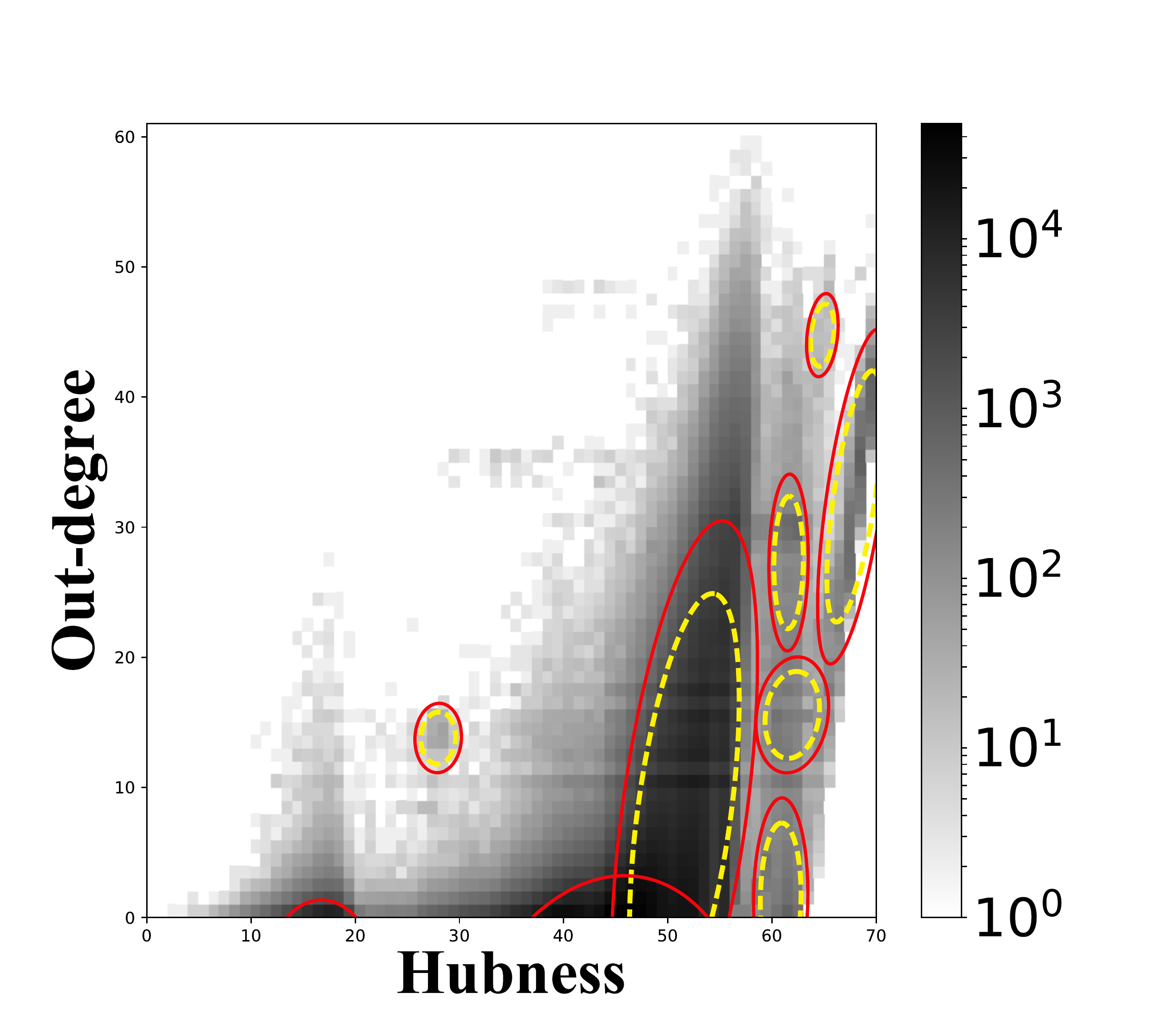}
		\vspace{-0.05in}
		\caption{Results from 
			\textbf{\Eaglemine}.}
		\label{fig:sina_outdhub_tnfit}
	\end{subfigure}  ~
	\begin{subfigure}[t]{0.375\linewidth} 
		\centering
		\includegraphics[height=1.2in]{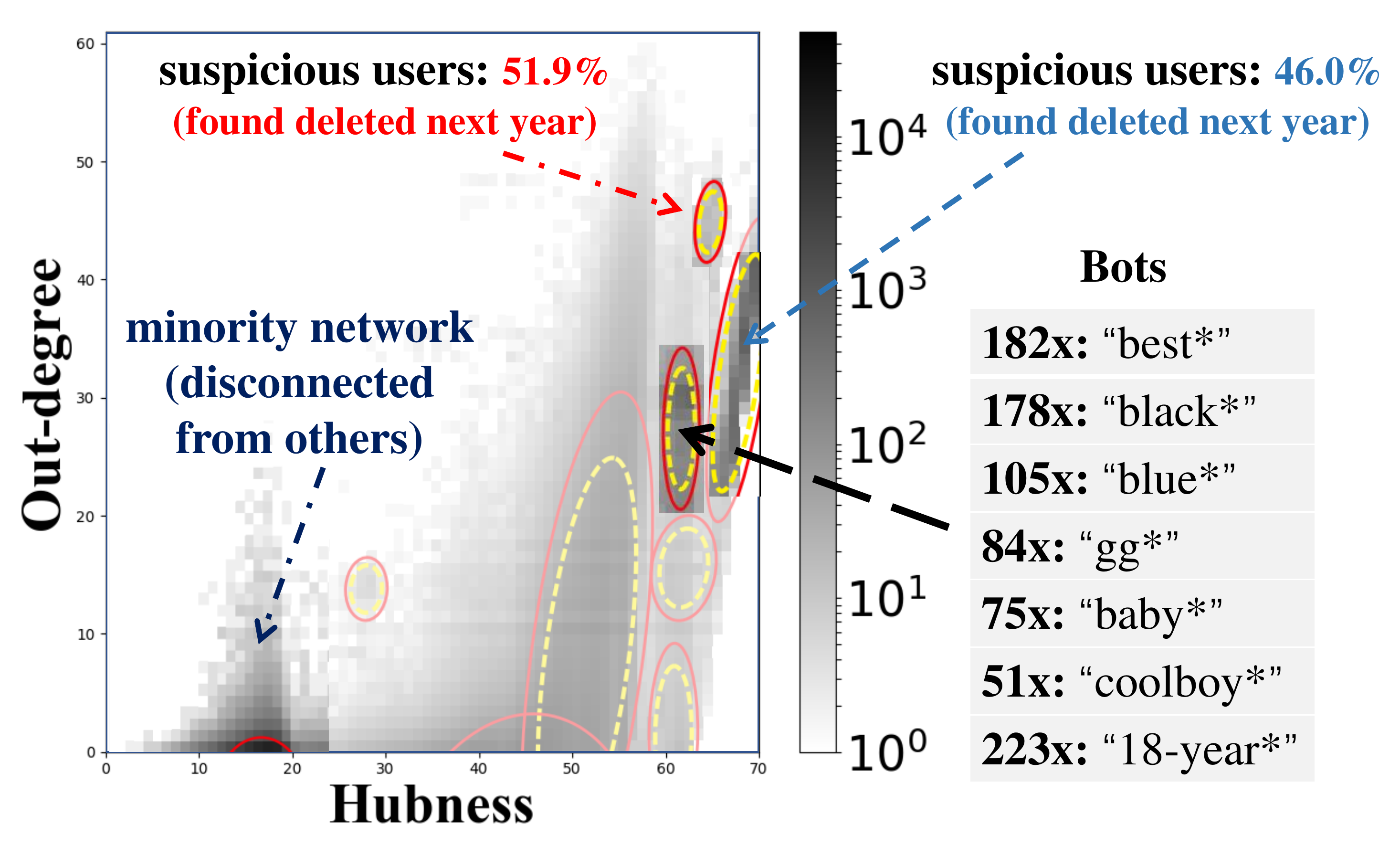}
		\vspace{-0.05in}
		\caption{Interesting and suspicious clusters.}
		\label{fig:sina_outdhub_cases}
	\end{subfigure}
	
	\begin{subfigure}[t]{0.16\linewidth}
		\centering
		\includegraphics[height=0.95in]{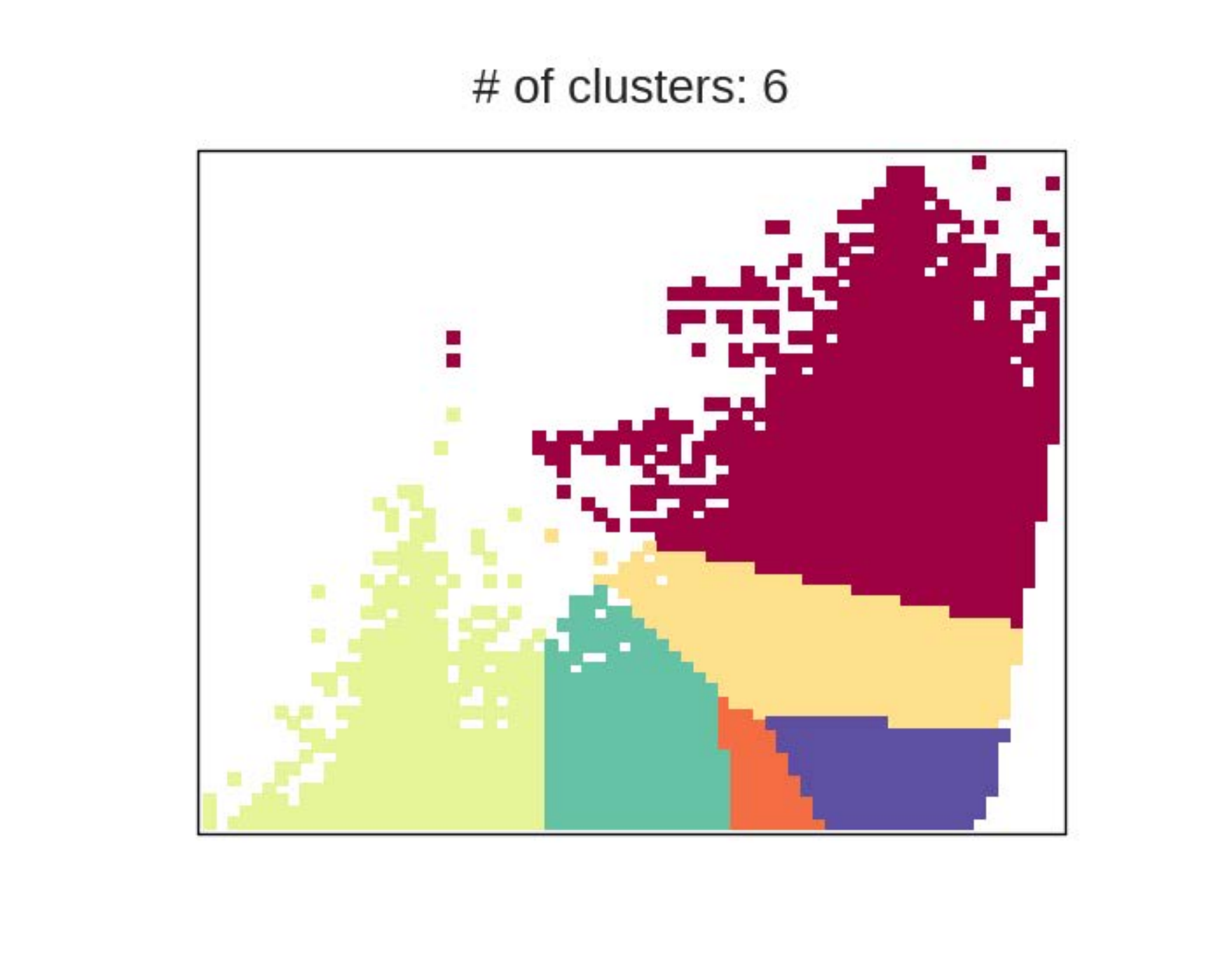}
		\caption{\centering{X-means}}
		\label{fig:wbxm}
	\end{subfigure} ~
	\begin{subfigure}[t]{0.16\linewidth}
		\centering
		\includegraphics[height=0.95in]{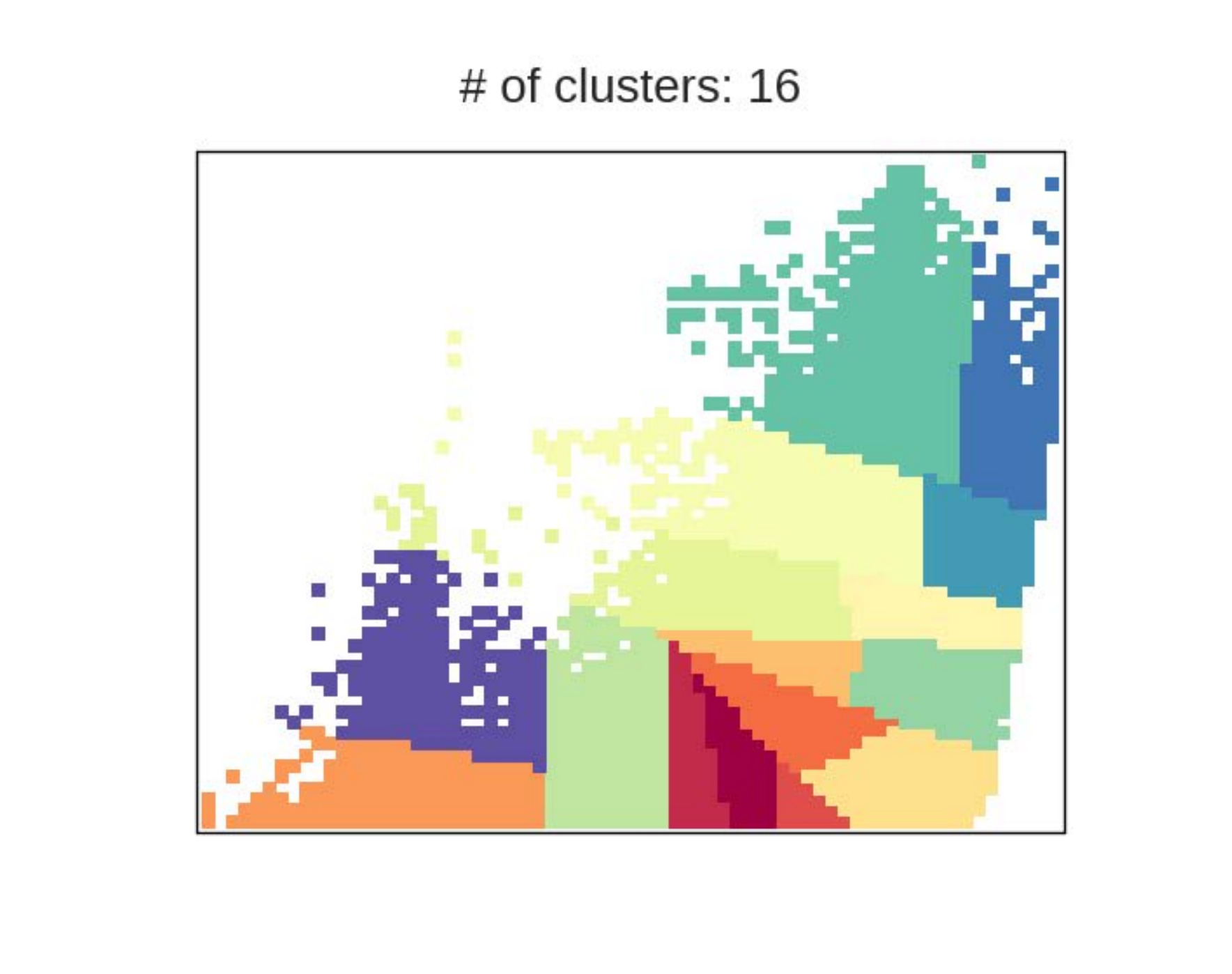}
		\caption{\centering{G-means}}
		\label{fig:wbgm}
	\end{subfigure} ~
	\begin{subfigure}[t]{0.16\linewidth}
		\centering
		\includegraphics[height=0.95in]{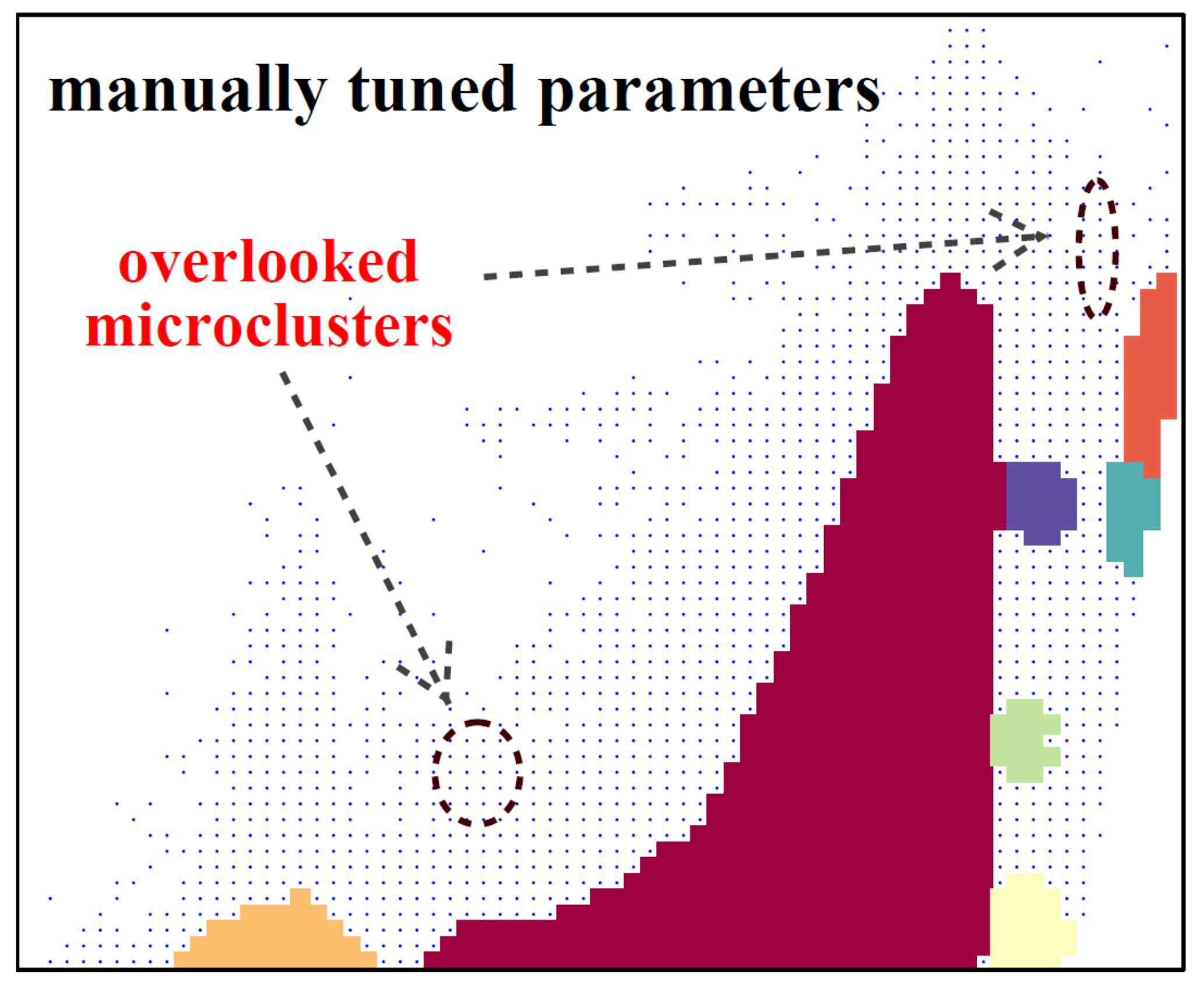}
		\caption{\centering{DBSCAN}}
		\label{fig:wbdbscan}
	\end{subfigure} ~
	\begin{subfigure}[t]{0.16\linewidth}
		\centering
		\includegraphics[height=0.95in]{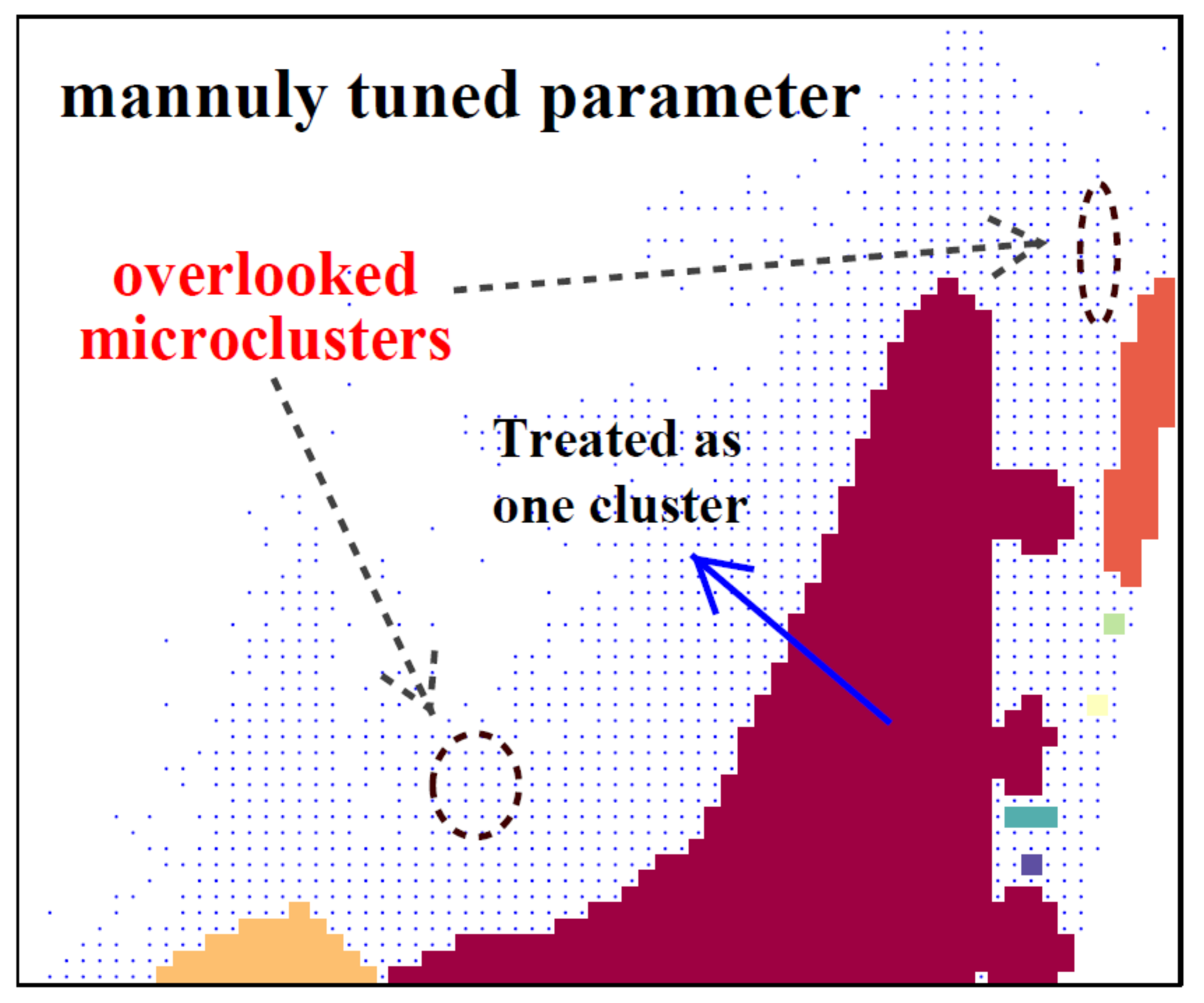}
		\caption{\centering{STING}}
		\label{fig:wbsting}
	\end{subfigure} ~
	\begin{subfigure}[t]{0.22\linewidth}
		\centering
		\includegraphics[height=0.95in]{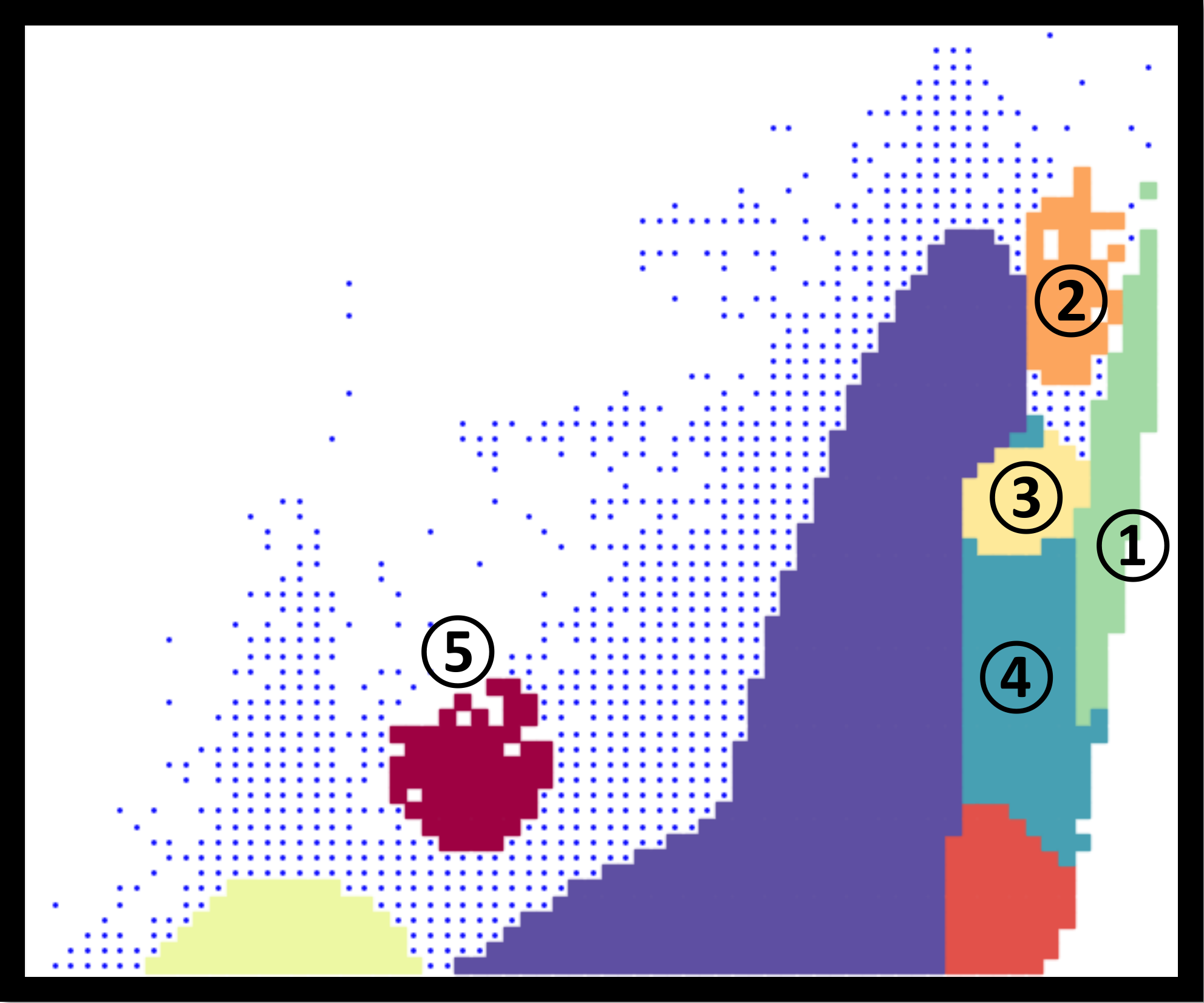}
		\caption{\centering{Our proposed \Eaglemine}}
		\label{fig:wbeaglemine}
	\end{subfigure}
    \vspace{-0.07in}
	\caption{ Our proposed \eaglemine achieves effective results on
	microblog Sina Weibo data. 
 	    (a) the histogram plot in heat map.
		(b) the summarization of \eaglemine with truncated distributions.
		(c) suspicious users deleted from the Sina Weibo website, bots
		with common username prefixes, and an isolated
		network found by \Eaglemine. 182x:``best*'' indicates that 182
		bots have common username prefix ``best'' in the cluster.
		(d)-(h) are separately the clustering results of the baselines and \Eaglemine.
		The blue scatter points in (f)-(h) denote individual outliers.
		Even with extensive manual work to tune parameters, 
		DBSCAN and STING still overlook two low-density micro-clusters
		(dashed circles) which are easily spotted by human vision and our \Eaglemine,
		since the fixed density threshold. Moreover, STING fails to
		separate three different clusters (red) on the right side of the major cluster. 
	}
    \vspace{-0.1in}
	\label{fig:wbexp}
\end{figure*}

\section{Related work}
\label{sec:related}
For the Gaussian clusters, K-means, X-means\cite{pelleg2000x}, 
G-means\cite{Hamerly2004Learning}, and BIRCH\cite{Zhang:1996birch} 
(which is suitable for spherical clusters) algorithms suffer from 
being sensitive to outliers.
Density based methods, such as DBSCAN\cite{ester1996density} 
and OPTICS\cite{ankerst1999optics} are noise-resistant and 
can detect clusters of arbitrary shape and data distribution, 
while the clustering performance relies on density threshold 
for DBSCAN, and also for OPTICS to derive clusters from reachability-plot.
RIC\cite{bohm2006robust} enhances other clustering algorithms 
as a framework, using minimum description language as 
goodness criterion to select fitting distributions 
and separate noise.
STING\cite{wang1997sting} hierarchically merges 
grids in lower layers to find clusters with a given density threshold.
Clustering algorithms\cite{roerdink2000watershed}
derived from the watershed transformation\cite{vincent1991watersheds},
treat pixel region between watersheds as one cluster,
and only focus on the final results and ignores the hierarchical 
structure of clusters. 
Community detection algorithms\cite{lancichinetti2009community}, 
modularity-driven clustering, and cut-based 
methods\cite{schaeffer2007graph} usually cannot handle large graphs 
with million nodes or fail to provide intuitive and interpretable 
result when applying to graph clustering.

\begin{table}[htbp]
	\centering
	\caption{Comparison between \eaglemine  and other related algorithms
		(DBD $=$ Dense block detection).}
	\vspace{-0.11in}
	\resizebox{\columnwidth}{!}{
		\begin{tabular}{L||cccccccc||c}
			& \rotv{X-means~\cite{pelleg2000x}}
			& \rotv{G-means~\cite{Hamerly2004Learning}}
			& \rotv{DBSCAN~\cite{ester1996density}}
			& \rotv{BIRCH~\cite{Zhang:1996birch}}
			& \rotv{STING~\cite{wang1997sting}}
			& \rotv{RIC~\cite{bohm2006robust}} 
			& \rotv{DBD~\cite{hooi2016fraudar,Pei2005On}}
			& \rotv{\eaglemine }\\
			\hline                
			parameter free	       & \checkmark & \checkmark  &            &            &            & \checkmark  & \checkmark & \checkmark \\
			non-spherical cluster  &            & \checkmark  & \checkmark &            & \checkmark & \checkmark  & ?          & \checkmark\\
			anomaly detection      &            &             & \checkmark &            & \checkmark & \checkmark  & \checkmark & \checkmark \\
			summarization          &            &             &            & \checkmark & \checkmark & \checkmark  &            & \checkmark \\
			linear in \#nodes      &            &             &            & \checkmark & \checkmark &             &            & \checkmark
		\end{tabular}%
	}
	\label{tbrelworks}%
	\vspace{-0.5in}
\end{table}%

Supported by human vision theory, including visual saliency, 
color sensitive, depth perception and attention of vision 
system\cite{heynckes2016predictive}, 
visualization techniques\cite{ware1988color, Borkin2011EAV} and HCI tools
help to get insight into data\cite{tukey1985computer,  akoglu2012opavion}. 
\textsc{Scagnostic}\cite{tukey1985computer,buja1992computing}
diagnoses the anomalies from the plots of scattered points.
\cite{wilkinson2005graph} improves the detection by statistical 
features derived from graph-theoretic measures.
Net-Ray\cite{Kang2014NetRayVA} visualizes and mines adjacency matrices
and scatter plots of a large graph, and discovers some interesting 
patterns. 
In terms of anomaly detection of the graph,
\cite{prakash2010eigenspokes,jiang2014inferring} find communities 
and suspicious clusters in graph with spectral-subspace plots.
\textsc{SpokEn}\cite{prakash2010eigenspokes} considers the 
``eigenspokes'' pattern on EE-plot produced by pairs of eigenvectors 
of graphs, and is later generalized for fraud detection.
As more recent works, dense subgraph and subtensor have been proposed
to detect anomaly patterns and suspicious 
behavior\cite{kumar2010structure,Chen2010Dense, hooi2016fraudar}.
Fraudar\cite{hooi2016fraudar} proposes a densest
subgraph-detection method that
incorporates the suspiciousness of nodes and edges during optimization.

\eaglemine differs from majority of 
above methods as summarized in Table~\ref{tbrelworks}.
Our proposed method \eaglemine is the only one that matches
all specifications.

\section{Proposed Model}
\label{sec:method}
First, let a graph $\mathcal G = (\pmb V, \pmb E)$, where $\pmb V$ is 
the node set $\pmb V$ and $\pmb E$ is the edge set. 
$\mathcal G$ can be either homogeneous, such as friendship/following 
relations, or bipartite as users rating restaurants.
Then our problem is informally defined as follows.  
\begin{problem}[Informal Definition]
    \label{prob:def}
    \textbf{Given} a graph $\mathcal G$ 
    and characteristics of \pmb{V}, we want to
    \begin{compactenum}[\bfseries 1.]
		\item \textbf{Separate} the nodes sharing similar characteristics 
		         into groups.
		\item \textbf{Find} outliers which are small groups of nodes or
		    scattering nodes with different characteristics and deviate away from the rest,
		         if they exist.
	  \end{compactenum}
     \textbf{Optimize}: the goodness-of-fit of characteristics distribution, 
               and the consistency of human visual recognition.
\end{problem}

In our problem, we are free to use any 
sort of node features 
(characteristics) 
that can better visualize large graphs for manual pattern recognition, 
such as degrees, spectral quantities, number of triangles and 
average neighbor degree etc. 
To visualize feature similarity, we bucketize feature values 
and construct histogram by mapping nodes 
into corresponding histogram cells. 
The grid width can be selected by the approaches~ in~\cite{Wasserman:2006nonpara}.
Thus nodes with similar characteristics form 
dense clusters of cells, which facilitates finding node groups.

More generally, each graph node can be associated with $k$ features, and nodes
are located into the cells of a  $k$-dimensional histogram $\mathcal{H}$.
$k$ can be more than 2 even though it is not easy to visualize
a high-dimensional histogram. In our work, we introduce 
\eaglemine in a two-dimensional (2D) histogram for simplicity, 
which is easy to extend to the high-dimensional case.

Therefore, to separate nodes into groups and find outliers, 
our summarization model 
consists of  
\begin{compactitem}
    \item \textbf{Configurable vocabulary:} 
    distributions for describing dense clusters of histogram cells. 
    \item \textbf{Parameters:}
     \emph{vocabulary term} (distribution type) and
	 \emph{parameter configuration} of the distributions used for describing
	a cluster, and \emph{number of samples} i.e. nodes in each cluster.
\end{compactitem}
The configurable vocabulary can include all suitable distributions,
such as Uniform, Gaussian, Laplace, and exponential distributions.

\section{Our proposed method}

Our method is guided by the traits of 
human vision and cognitive system as follows:
\begin{trait}
    \label{traitisl}
     Human vision usually detects connected components, which 
     can be rapidly recognized by eyes despite substantial appearance 
     variation\cite{dicarlo2012neuron, liu2015understanding}. 
\end{trait}
This insight helps us to identify different connected and dense regions 
in the histogram as candidate clusters and makes refinement for 
smoothing.
\begin{trait}
    \label{traithie}
     Top-to-bottom recognition and hierarchical 
     segmentation\cite{arbelaez2011contour}. Humans organize 
     basic elements (e.g. words, shapes, visual-areas) into 
     higher-order groupings to generate and represent
     complex hierarchical structures in human cognition and 
     visual-spatial domains.
\end{trait}
This trait suggests organizing and exploring node groups based on 
a hierarchy, or tree structure, as we will do.

\subsection{Water-level tree}
\label{sect:water-level tree}

First, let histogram $\mathcal H = \{h_{i, j}\,|\,i\in[1, n_r]; j\in[1,n_c]\}$, 
where $i$, $j$ are row and column index of histogram grids respectively, 
and height $h_{i, j}$ is the number of nodes in the grid $(i, j)$. 
We name a connected region formed by non-empty grids ($h>0$) as an
\emph{\textbf{island}}, and the regions formed by empty
grids as \emph{\textbf{water area}}.
Assume that we can flood island areas which makes grids with
$h_{i, j}<r$ to be underwater, 
i.e. setting $h_{i, j} = 0$, where $r$ is a water level. 
Then we call the resulting islands as the islands in the condition 
of water level $r$.

To organize all found islands during the flooding of different 
water-levels for \pltname $\mathcal{H}$, we propose to construct a 
water-level tree $\mathcal{T}$. 
The nodes represent islands and the edges denote the relationship 
that child islands come from the parent island owing to the water flooding. 
Thus, the root of $\mathcal{T}$ corresponds water-level zero, 
while moving towards the leaves, the nodes are at a higher 
water-level of flooding.

In fact, the islands are candidate clusters as Trait~\ref{traitisl}. 
The flooding process intuitively reflects how human eyes hierarchically 
capture different objects from the color histogram $\mathcal H$ 
as Trait \ref{traithie} (see Figure\ref{fig:sina_outdhub}),
where the gradient colors depict clusters at different water-levels.

The \textsc{WaterLevelTree} algorithm is shown in Algorithm~\ref{alg:islds}.
We start from the root with all non-empty grids.
Rising the water-level by threshold $r$ represents 
scanning different density of grids.
Due to the power-law-like distribution of grid heights,
we use the logarithm of heights. 
The water-level increases from $r=0$ to $\log h_{max}$ with 
a fixed step size $s$, where $h_{max} = \max{\mathcal H}$.
For smoothing the islands at each eater-level we use binary 
opening ($\circ$), a basic workhorse of mathematical morphology,
to remove small objects (noise) from the foreground of the 
histogram and also separate weakly-connected parts with a specific structure 
element\footnote{Here we use 2$\times$2 square-shape ``probe''.}.
Afterwards, we link each island at current level $r$ to its parent 
at level $r-1$ of the tree (see Figure\ref{fig:contract}).
Finally, when $r = \log h_{max}$, very small grids may connect 
to the tree $\mathcal{T}$.

\begin{algorithm}[thb]
	\caption{\textsc{WaterLevelTree} Algorithm}
	\label{alg:islds}
	\begin{algorithmic}[1]
		\Require \pltname $\mathcal{H}$. 
		\Ensure Water-level tree $\mathcal{T}$ of islands.
		\State $\mathcal{T} =$ non-empty grid set in $\mathcal{H}$ as root.
		\For{water level $r = 0$ to $\log{h_{max}}$ by step $s$}
		\State $\mathcal{H}^{r} =$ assign $h_{i,j}\in \mathcal H$ to 
		zero if $\log h_{i,j} < r$.
		\State $\mathcal{H}^{r}=\mathcal{H}^{r}\circ${} \textbf{E}. 
		\Comment{\textit{binary opening} smooth}. 
		\State islands $\Theta^r =$ connected regions in $\mathcal{H}^{r}$. 
		\State link each island in $\Theta^r$ to their parents in $\mathcal T$.
		\EndFor
		\State Contract $\mathcal{T}$ by iteratively removing the single-child 
		islands and linking its children to parent.
		\State Prune $\mathcal T$ heuristically to remove noise islands.
		\State Expand island $\theta \in \Theta$. 
		\State \Return $\mathcal{T}$
	\end{algorithmic}
\end{algorithm}

The current tree $\mathcal T$ may contain many nodes with only one 
child, meaning no actual new islands appear, which leads to many 
redundant nodes due to the increasing step $s$ of water-level.
Therefore, we propose to use the following three steps to post-process 
the tree as refinement: \emph{contract, prune and expand}.

\textbf{Contract:}
We search the tree using depth-first search (DFS), and if a 
single-child node is found, we remove it and link its 
children to its parent. 
In the sequel, all single-child nodes will be contracted to 
their parents.
This is shown in Figure~\ref{fig:contract}, where the dashed 
lines with arrows depict that the single-child's children 
are linked to its parent. 

\textbf{Prune:}
The purpose of pruning is to smooth an island which has noise 
peaks on top of the island due to fluctuations of grid heights.  
As the example at bottom right of Figure\ref{fig:prune} shown, 
the island $\theta$ is at water-level $r$. When water-level 
rises to $r'$, three small peaks as noise islands connect to 
their parent $\theta$, so they are needed to be removed for smoothing.
Thus, with a breadth-first search (BFS), 
we prune such child branches (including the children's descendants)
according to the heuristic rule that the total area  of child islands
is less than half of their parent node
(see \textcircled{1} and \textcircled{2} in Figure\ref{fig:prune}).

\textbf{Expand:}
For the sake of fitting an island with a vocabulary term,
we try to include additional surrounding grids to avoid 
over-fitting in parameter learning. 
So we expand each island with its neighboring grids under the 
water, and iteratively absorb outward non-empty neighbors of 
each island until overlapping with other islands，
or doubling the area of the original island. 
The expansion results are illustrated with shadowed rings in 
Figure\ref{fig:prune} (see \textcircled 3).

\begin{figure}[t]
\centering
	\begin{subfigure}[t]{0.40\linewidth}
	   \centering
		\includegraphics[height=1.4in]{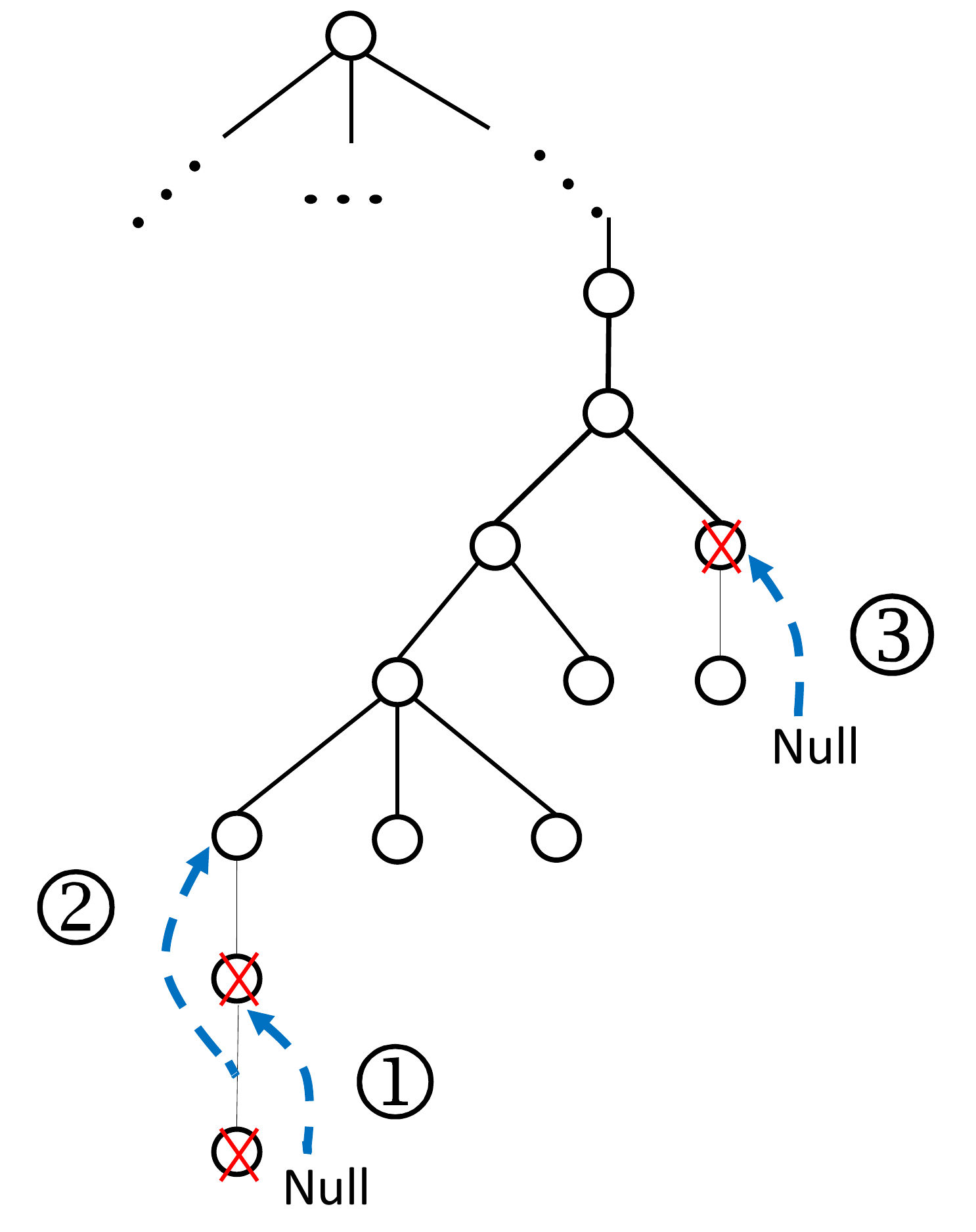}
      \caption{Tree contract}
		\label{fig:contract}
   \end{subfigure}
    \vline ~
   \begin{subfigure}[t]{0.46\linewidth}
	   \centering
		\includegraphics[height=1.4in]{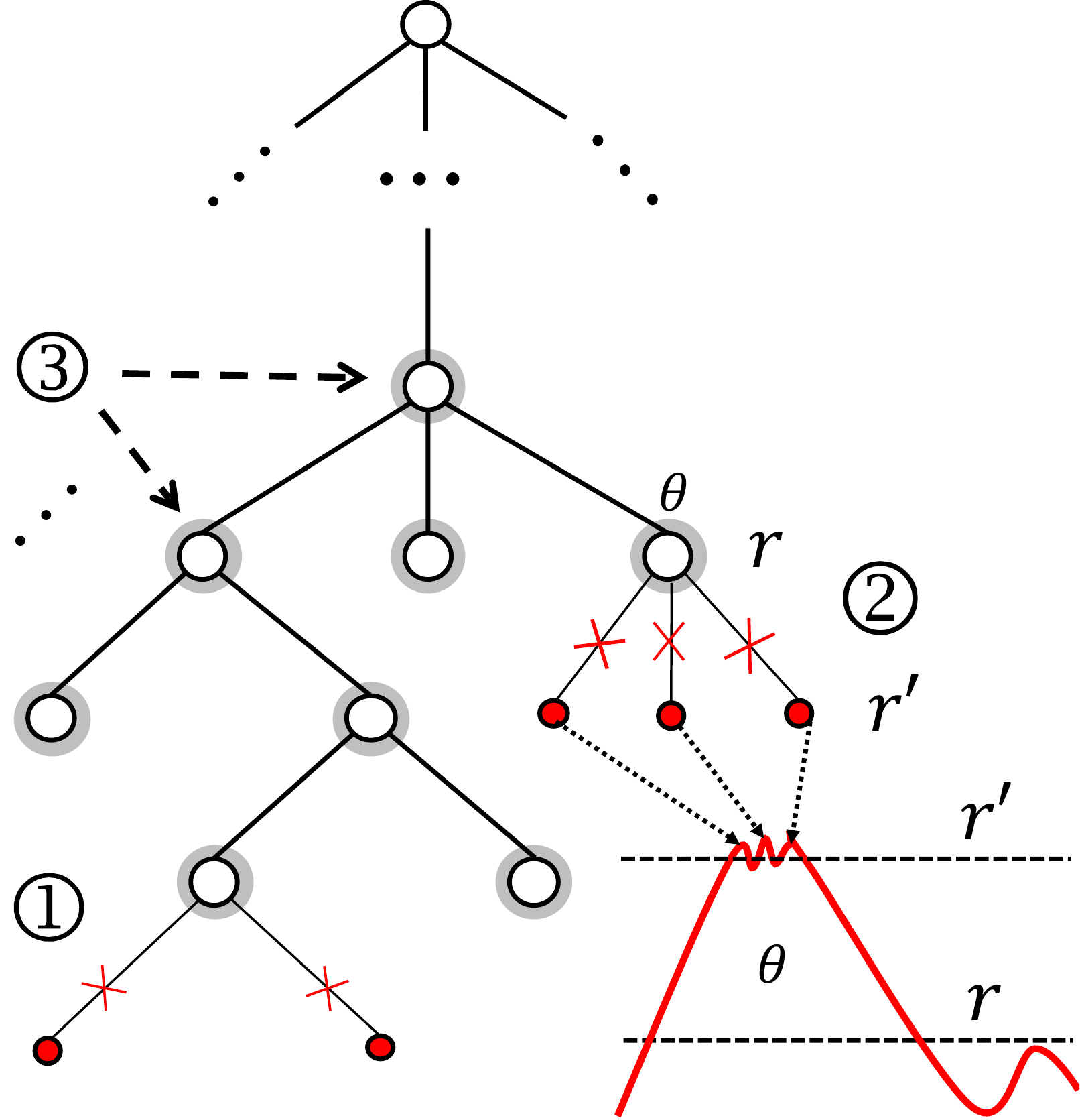}
      \caption{Tree prune and node expand}
		\label{fig:prune}
   \end{subfigure}
  
   \begin{subfigure}[t]{0.9\linewidth}
	   \centering
		\includegraphics[height=1.4in]{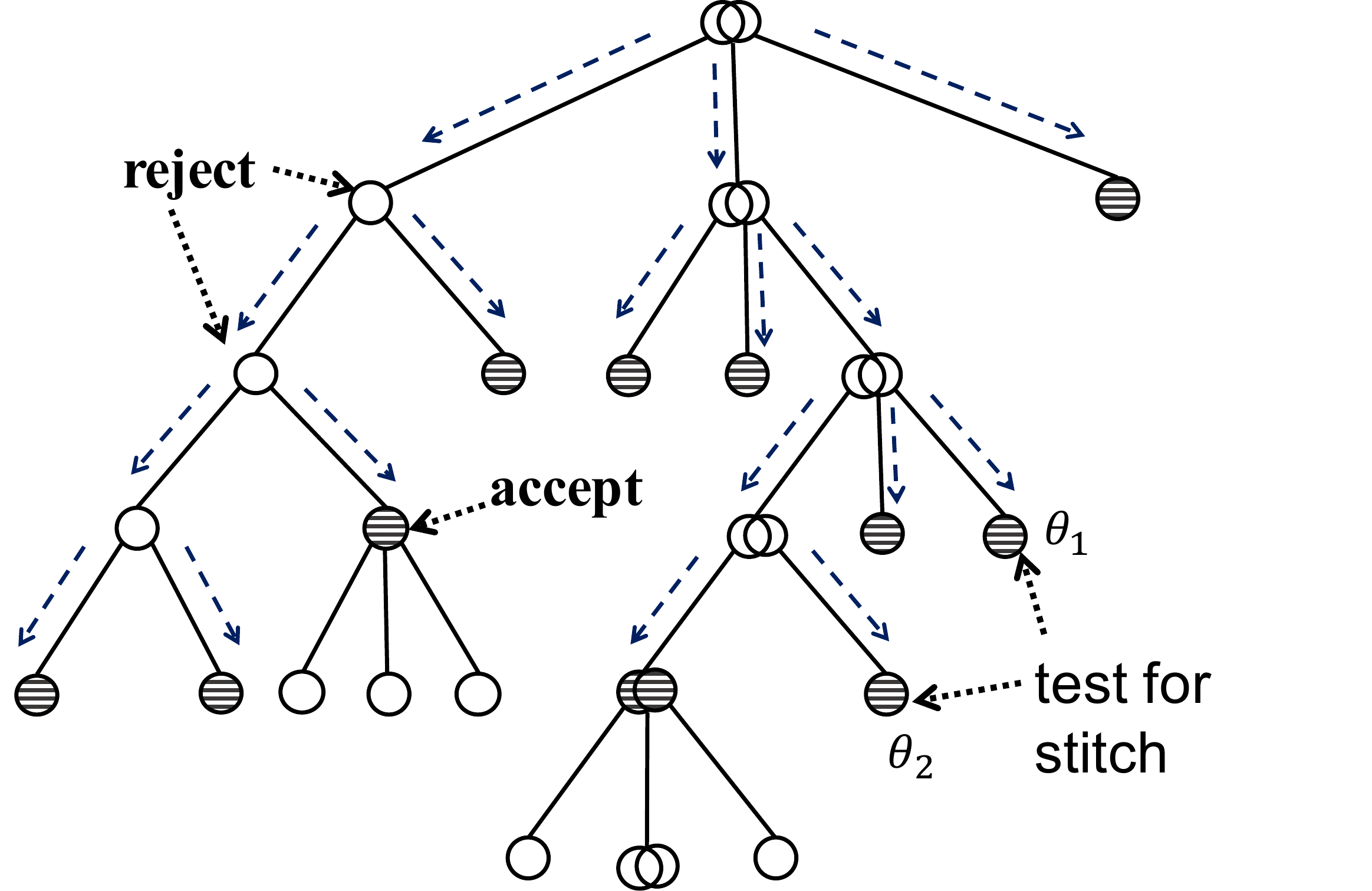}
      \caption{Statistical hypothesis Test and optimal islands search}
		\label{fig:hyp_test}
   \end{subfigure} 
    \vspace{-0.1in}
    \caption{Key steps in proposed 
    algorithms. }
    \vspace{-.21in}
\end{figure}

It is worth mentioning that in the  
watershed\cite{vincent1991watersheds,roerdink2000watershed} 
formalization, the non-empty areas of $\mathcal H$ are defined 
as catch basins for clustering purpose, and watersheds form 
the boundaries between clusters. 
These clusters actually correspond to the leaf islands 
in our water-level tree. 
However, \eaglemine builds the whole hierarchy of islands 
at different water-levels. 
We shall see later that \eaglemine searches the tree to find a better 
combination of clusters with hypothesis test and describes
the results with vocabulary rather than only keeping leaf islands.

Another tree-based clustering algorithm STING\cite{wang1997sting}
builds a multi-resolution tree for the histogram grids 
as an index for quick querying. The clusters were directly 
achieved by querying the tree-like index with density threshold 
$c$ as a parameter, while those clusters are actually the islands 
in the same level of water-level tree.

\subsection{Describing islands}
In the model, we define a  configurable vocabulary for describing
islands, which is flexible to include any user-defined distributions.
As we know, many features of a graph, 
such as node degrees, $k$-cores, and triangles,
typically follow power-law distributions, 
which makes some grids along the histogram boundary have 
the highest density (see the bottom region of Figure\ref{fig:sina_outdhub}).
Such an island can be described with a Gaussian 
distribution truncated at the boundary.
Moreover, the truncated normal distribution is powerful 
enough to fit various sizes and shapes of islands, 
e.g. lines, circles, ellipses, truncated ellipses. 
Moreover, grids in histogram $\mathcal{H}$ are discrete units, 
so we propose to select $digitized$, $truncated$ and 
$multivariate$ Gaussian distribution (DTM Gaussian for short) 
as one of the configurable vocabulary. 
The DTM Gaussian distribution is defined as:
\begin{definition}[DTM Gaussian distribution]
	The probability function of digitized, truncated, and multivariate Gaussian is
	$$
	\setlength{\abovedisplayskip}{1pt}
	\setlength{\belowdisplayskip}{1pt}
	P(g;\pmb{\mu, \Sigma}, \pmb{a}, \pmb{b}) = 
	\int\!\!\cdots\!\!\int_{g} \psi(\pmb{x}; \pmb{\mu, \Sigma}, \pmb{a},
	\pmb{b}) d\pmb x
	$$
	where $g$ is the grid for digitization in histogram, and the
	probability is an integral over hypercube $g$.
\end{definition}
where $\psi(\pmb{x}; \pmb{\mu, \Sigma}, \pmb{a}, \pmb{b})$ is a 
2-dimensional truncated normal distribution with
the mean $\pmb{\mu}$, non-singular covariance matrix $\pmb\Sigma$,
and truncated at $\pmb{a} = [0, 0]^T$, 
$\pmb{b}=[+\infty, +\infty]^T$ 
for lower and upper bound respectively.

In addition, we choose the overlapped DTM Gaussian distributions 
mixed with equal-weight as another vocabulary term to capture 
the skewed triangle-like island as shown in Figure\ref{fig:sina_outdhub}.
In our data study, this triangle-like island exists in many different 
histogram plots, and it contains the majority of graph nodes.
Taking the user-retweet-message bipartite graph as an example, 
Figure~\ref{fig:sina_outdhub} depicts users' distribution over 
out-degree and hubness (the first left singular vector).
The power-law makes the density decrease along the vertical axis, 
and users with the same degree show normal distribution
horizontally due to retweeting various types of messages 
with different hubness.

Our model can include any predefined distributions n vocabulary. 
We can use distribution-free statistical hypothesis test, 
like Pearson's $\chi^2$ test, 
or other heuristic-based approaches (as shown later) 
to decide the optimal distribution for each island in practice.
With the above vocabulary, we use maximum likelihood estimation 
to learn the parameter configuration for each term.
To learn the number of samples $N$, we minimize 
the total expectation error of all grids in an island:
$\sum_{g}\norm{N \cdot P(g;\cdot) - h}$.
We can choose the size of training data as $N$ directly
since this optimizes log-likelihood.

\subsection{\eaglemine Algorithm}
\label{sec:searchtree}
The overall view of our \eaglemine algorithm is described 
in Algorithm~\ref{alg:eaglemine}.
\eaglemine hierarchically detects \mcls in the \pltname 
plot $\mathcal H$, and outputs the parameters $\pmb C$ of
vocabulary term (e.g. single or mixture of DTM Gaussian),
parameter configuration, and number of samples for summarization. 
The key steps are shown as follows. 

First, the water-level tree $\mathcal T$ is constructed as 
Algorithm~\ref{alg:islds}.
We then decide the configurable term $y_{\theta}$ for
each node $\theta$.
In particular, we choose the mixture DTM Gaussian for a node 
depending on two observations:
1) the majority of normal nodes located at large triangle-like dense area;
2) one of the children should inherit the mixture distribution 
from the parent since children are part of their parent island 
in lower water level of $\mathcal T$. 
Thus, we get the vocabulary term $y_{\theta}$ for each 
island node, 
which is either DTM Gaussian or its mixture.

Afterwards, we search along the tree $\mathcal T$ with BFS to 
select the optimal combination of clusters (see step 5 to 14), 
based on the statistical hypothesis test 
(Anderson-Darling test\cite{stephens1974edf}).
$DistributionFit(\theta, y_{\theta})$ is used to denote the process 
of learning parameters for term $y_{\theta}$.

\begin{algorithm}[t] 
	\caption{\eaglemine Algorithm}
	\label{alg:eaglemine}
	\begin{algorithmic}[1]
		\Require \pltname $\mathcal{H}$ for node features of graph $\mathcal{G}$.
		\Ensure summarization $\pmb{C}$ consisting of distribution type, 
		        fitting parameters and number of samples for each distribution. 
		\State $\pmb{C} = \emptyset$
		\State $\mathcal{T} =$ \textsc{WaterLevelTree}$(\mathcal H)$
		\State $Y =$ decide distribution type $y_{\theta}$ from vocabulary 
		       for each island in $\mathcal{T}$. 
		\State queue $\pmb{Q} =$ root node of $\mathcal{T}$.
		\While{ $\pmb{Q} \ne \emptyset$} \hfill $\triangleright$ breath-first search (BFS)
			    \State ${\theta} \leftarrow$ dequeue of $\pmb Q$.
		\State $c = DistributionFit(\theta, y_{\theta})$
		\State Hypothesis test $\pmb H_0 = $ grids of island ${\theta}$ 
		        come from distribution $c$.
		\If {$\pmb H_0$ is rejected} 
		\State enqueue all the children of ${\theta}$ into $\pmb{Q}$
		\Else
		\State add $c$ into $\pmb{C}$
		\EndIf
		\EndWhile
		\State Stitch promising distributions in $\pmb{C}$.
		\State \Return summarization $\pmb{C}$.
	\end{algorithmic}
\end{algorithm}

Starting from the root node of $\mathcal{T}$,
for each node $\theta$ in search queue $\pmb{Q}$, 
we fit it with $c$ in configurable vocabulary. 
To decide whether to continue the BFS search,
we then assume the following null hypothesis:
\begin{equation}
	\setlength{\abovedisplayskip}{1.5pt}
	\setlength{\belowdisplayskip}{1.5pt}
    \nonumber
    \pmb H_0\text{: the grids of island $\theta$ come 
    	from distribution $c$.}
\end{equation}
Due to the variability of the grid-height in each island,
we have tried Pearson's $\chi^2$ test, BIC and AIC criteria, 
but the extreme heights made the test and other criteria unstable.
Thus, we test an island based on its binary image,
which focuses on whether the island's shape looks like 
a truncated Gaussian or mixture. 
In performing the hypothesis test, we project the grids on 
the major and minor axis separately in 2D space, 
and we accept the null hypothesis only when $\pmb H_0$ 
is true for both orientations.
So one-dimensional Anderson-Darling test (1\% significant level) 
is conducted on both axes,
and if one of the tests is rejected, the null hypothesis
will be rejected.

If $\pmb H_0$ is not rejected, we stop searching the island\rq{}s 
children and use current parameter $c$ to describe 
the island $\theta$. Otherwise, we need to further
explore the descendants of $\theta$.
The dashed lines with an arrow in Figure\ref{fig:hyp_test} 
demonstrate this process.
The final optimal combination of islands is 
shown in circles with texture.

Furthermore, some tiny small islands in the results $\pmb C$ 
may come from different parents (e.g.~ $\theta_1$ and $\theta_2$ in
Figure\ref{fig:hyp_test}). 
In such case, those islands that are physically close to each 
other may potentially be summarized by one distribution from the vocabulary. 
Thus, we use \emph{stitch} process in step 15 to merge them.
In the same way, we perform Anderson-Darling test to every pair of 
promising islands until no more changes occur.
When there are multiple pairs of islands that be merged at the same time,
we heuristically choose the islands pair with the least average 
log-likelihood reduction after stitching:
$$
\setlength{\abovedisplayskip}{1pt}
\setlength{\belowdisplayskip}{1pt}
(\theta_{i^*}, \theta_{j^*}) = 
    \argmin_{i, j}\frac{\mathcal L_i + \mathcal L_j - \mathcal L_{ij}}
                       {\text{\#points of $\theta_{i}$ and $\theta_j$}}
$$
where $\theta_i$ and $\theta_j$ are the pair of islands to be merged,
$\mathcal L_{(\cdot)}$ is log-likelihood of a island, 
and $\mathcal L_{ij}$ is the  log-likelihood of the merged island.

At last, the summarization $\pmb{C}$ of the 
histogram $\mathcal H$ is returned. 
Moreover, the histogram grids with very low probability ( $< 10^{-5}$) 
from all distributions are outliers.

Furthermore, since the only one main island fitted 
by mixture distribution contains the majority and normal nodes, 
we calculate the suspiciousness of 
other islands by the KL-divergence of their distributions:
$
    \sum_g N_i\cdot \text{KL}(P_{\theta_i}(g)\, ||\, P_{\theta_m}(g)),
$
where $\theta_m$ is the main island, $P(g)$ is the probability in grid $g$, and 
$N_i$ is the number of samples from the distribution of island $\theta_i$.

\subsection{Time complexity}
Given the features associated with nodes $\pmb V$,  
the time complexity for generating histogram $\mathcal{H}$ is 
$O(|V|)$. 

Let $M$ be the number of non-empty grids in the 
\pltname $\mathcal{H}$, and $K$ be the number of clusters.
We use gradient-descent to learn parameters in 
$DistributionFit(\cdot)$ of \eaglemine algorithm, 
so we assume that the number of iterations is $T$, 
which is related to the differences between initial and 
optimal objective values.
Then we have:
\begin{theorem}
    The time complexity of \eaglemine algorithm is 
    $O(\displaystyle\frac{\log h_{max}}{s}\cdot M + K\cdot T \cdot M)$.
\end{theorem}

\begin{proof}
	See the supplementary document.
\end{proof}

\section{Experiments}
\label{sec:experiments}

\begin{figure*}[thp]
	\centering
	\begin{subfigure}[t]{0.145\linewidth}
		\centering
		\includegraphics[height=0.65in, width=1.0in]{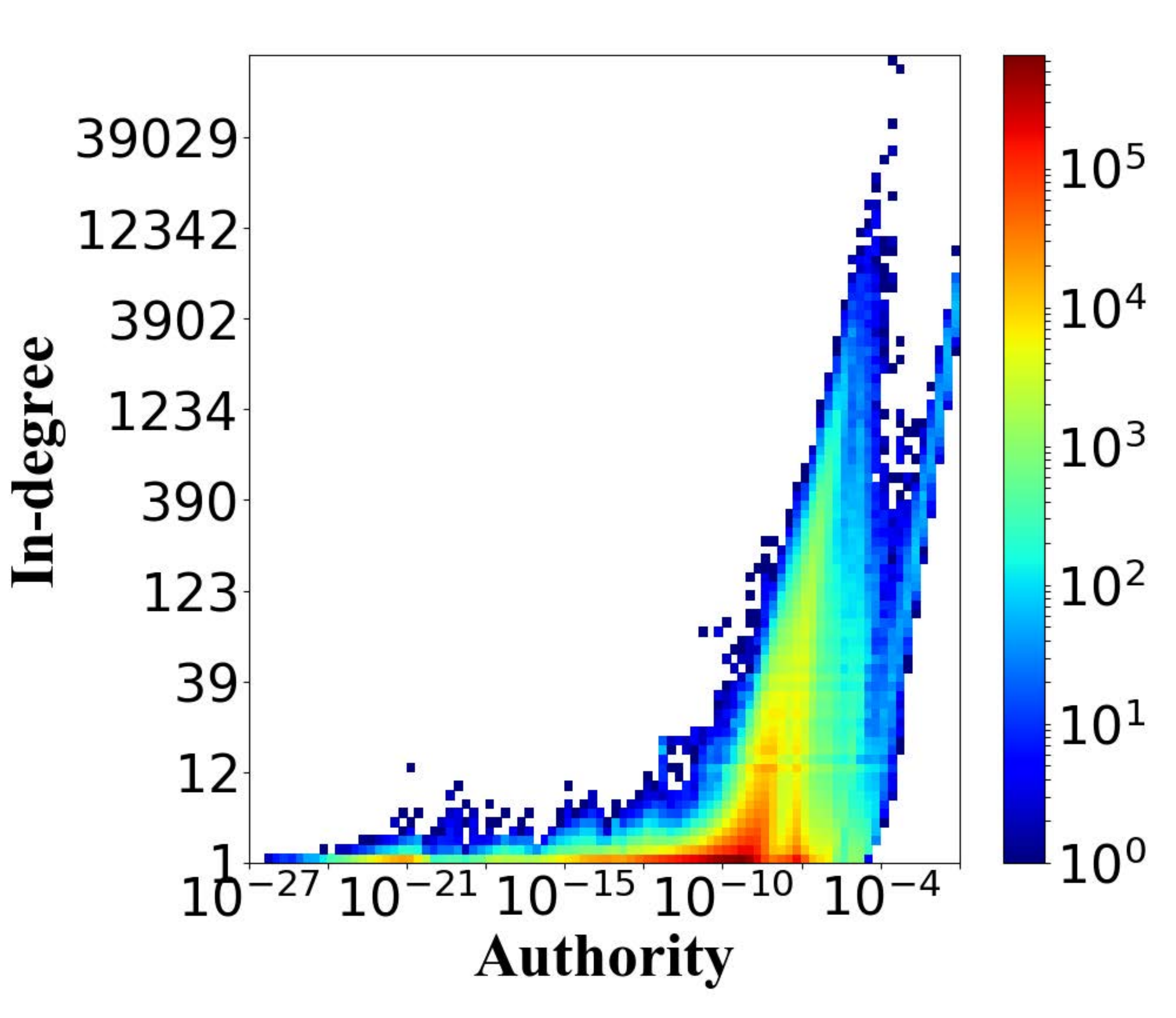}
		\caption{\centering{In-degree \,\, vs. Authority}}
		\label{fig:sinanov_inaut}
	\end{subfigure}
	\begin{subfigure}[t]{0.11\linewidth}
		\centering
		\includegraphics[height=0.65in]{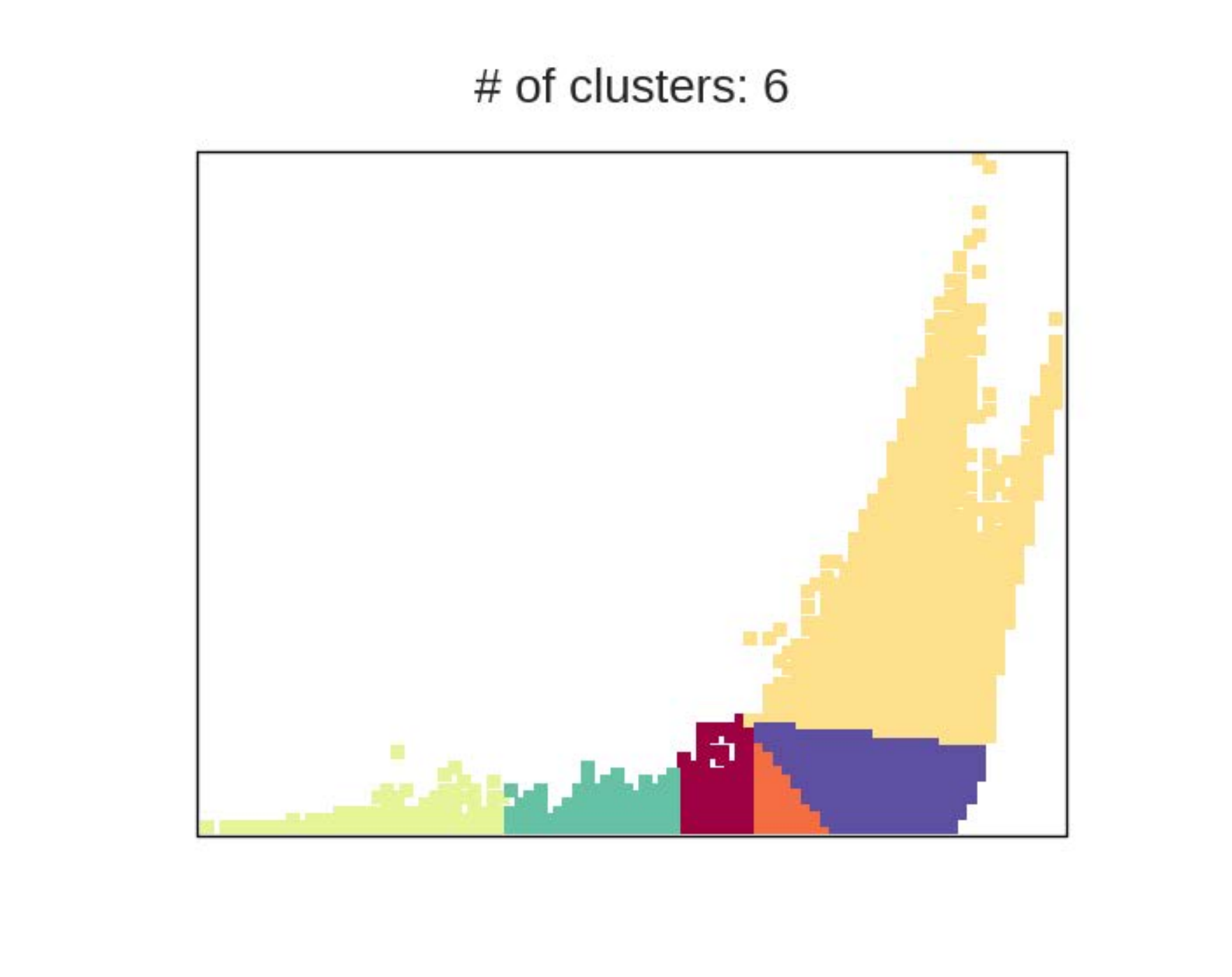}
		\caption{\centering{X-means \{6\}}}
		\label{fig:sinanov_inaut_xmeans}
	\end{subfigure}
	\begin{subfigure}[t]{0.11\linewidth}
		\centering
		\includegraphics[height=0.65in]{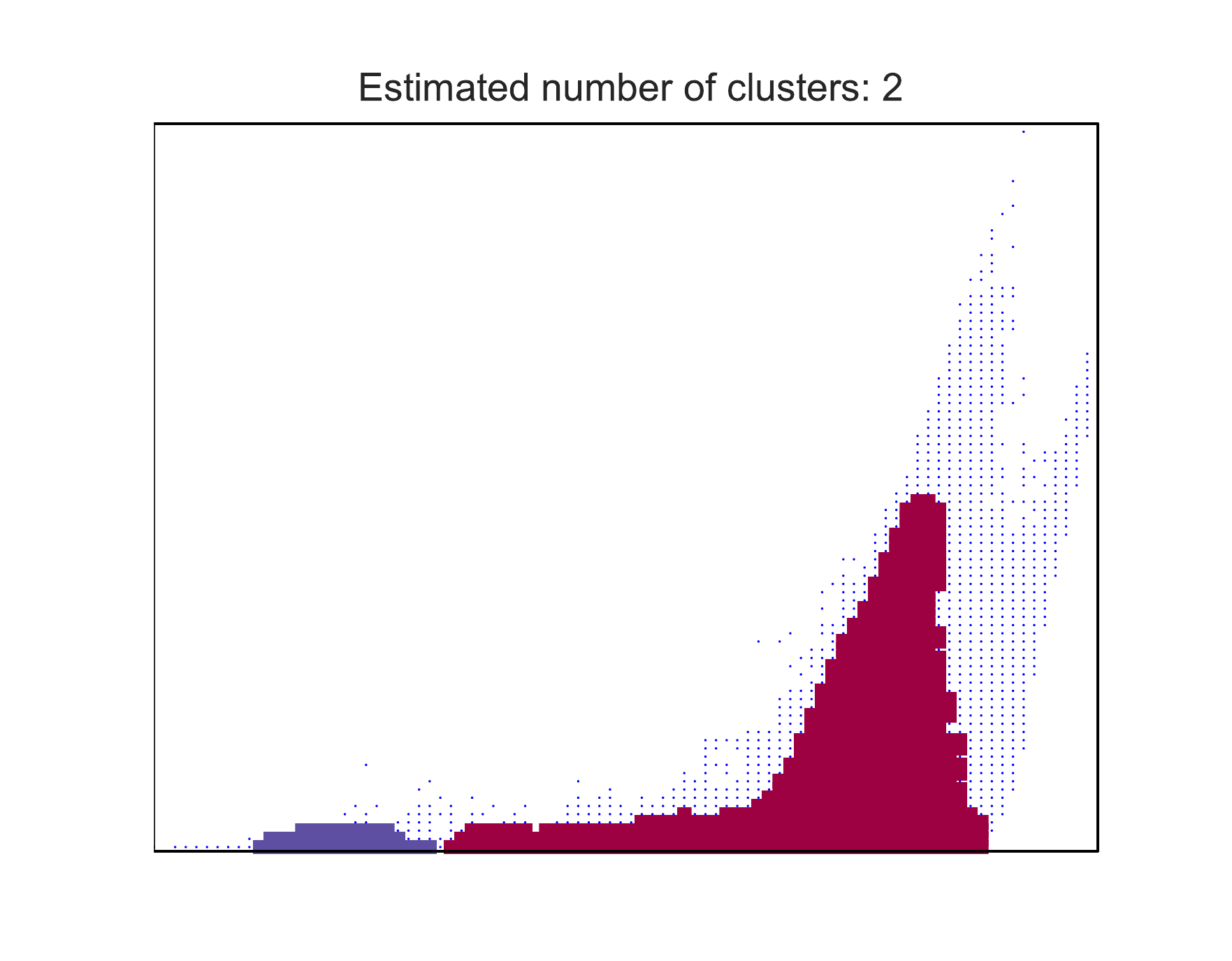}
		\caption{\centering{sting (\textbf{manual}) \{2\}}}
		\label{fig:sinanov_inaut_dbscan}
	\end{subfigure}
	\begin{subfigure}[t]{0.11\linewidth}
		\centering
		\includegraphics[height=0.65in]{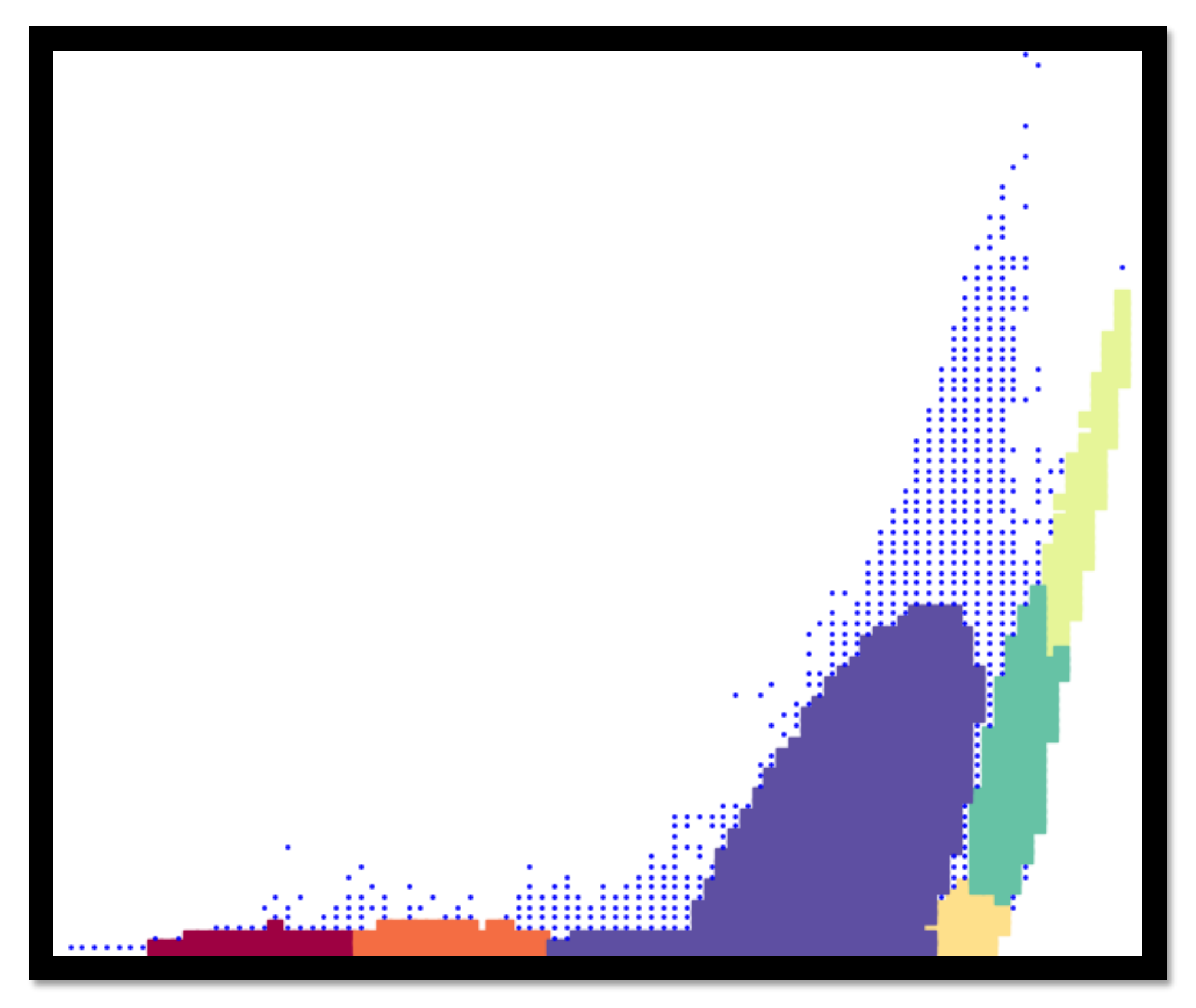}
		\caption{\centering{\eaglemine  \{6)\}}}
		\label{fig:sinanov_inaut_eaglemine}
	\end{subfigure}
	\vline
	\begin{subfigure}[t]{0.145\linewidth}
		\centering
		\includegraphics[height=0.65in, width=1.0in]{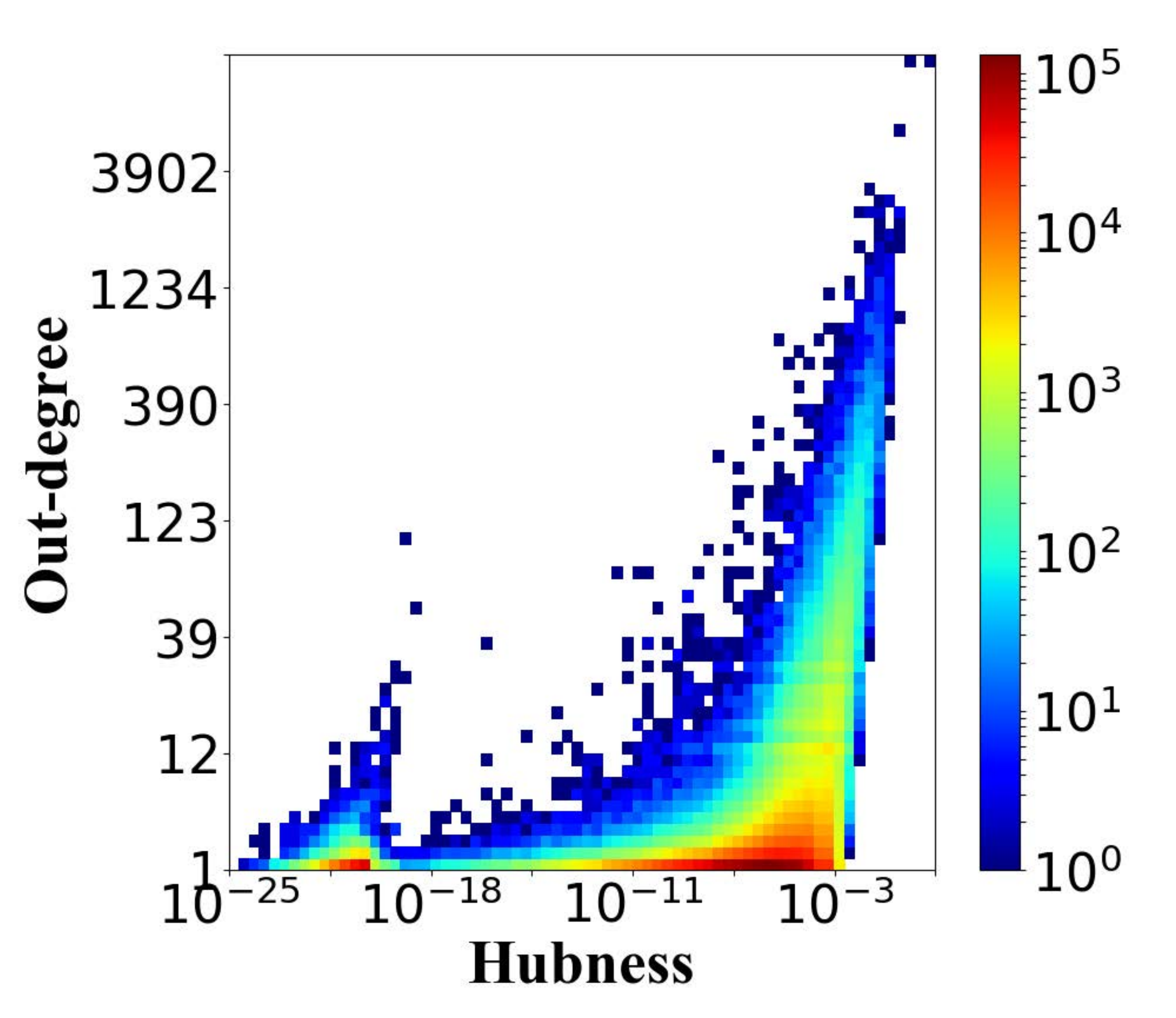}
		\caption{\centering{Out-degree \, vs. Hubness}}
		\label{fig:amazon_outhub}
	\end{subfigure}
	\begin{subfigure}[t]{0.11\linewidth}
		\centering
		\includegraphics[height=0.65in]{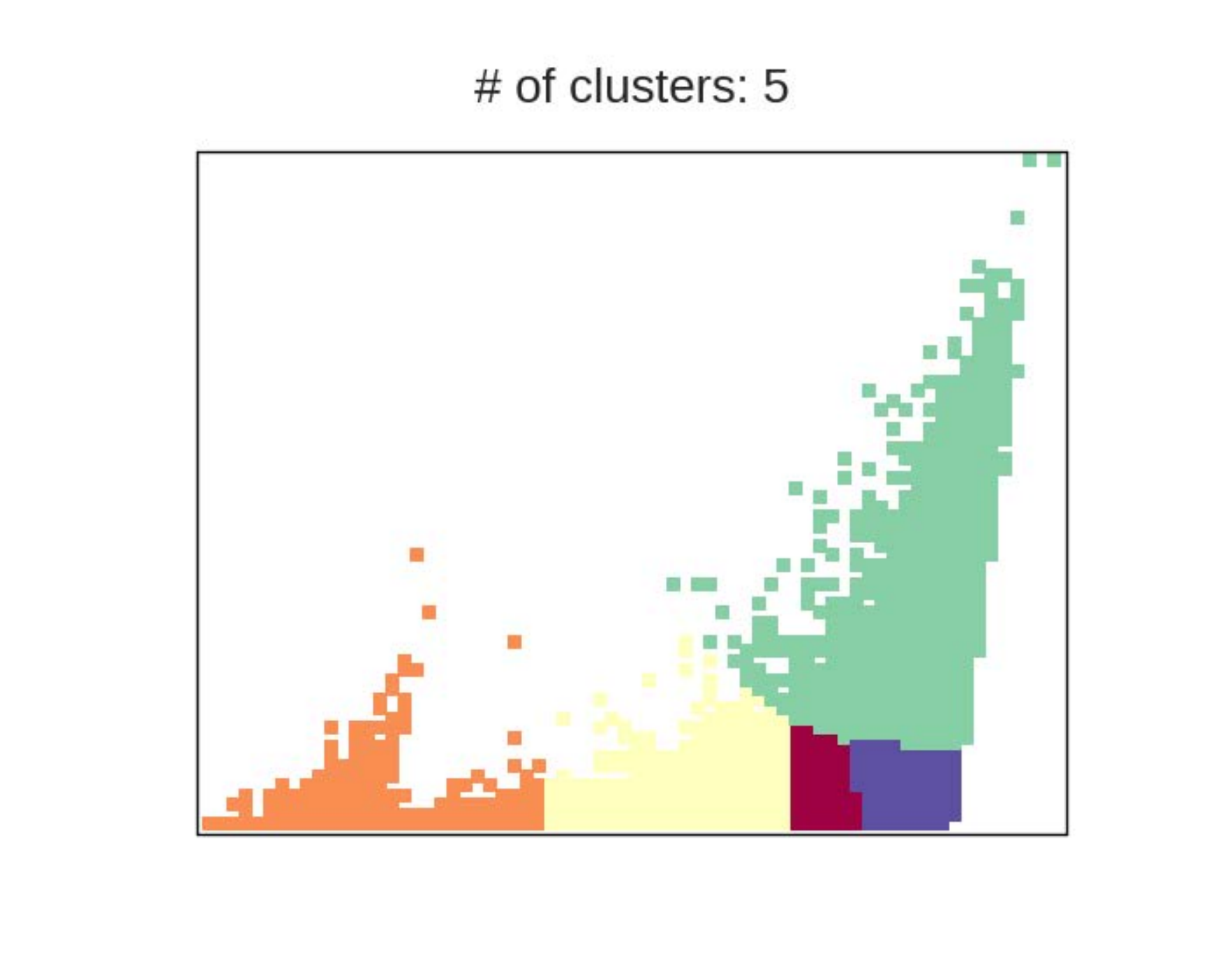}
		\caption{\centering{X-means \{5\}}}
		\label{fig:amazon_outhub_xmeans}
	\end{subfigure}
	\begin{subfigure}[t]{0.11\linewidth}
		\centering
		\includegraphics[height=0.65in]{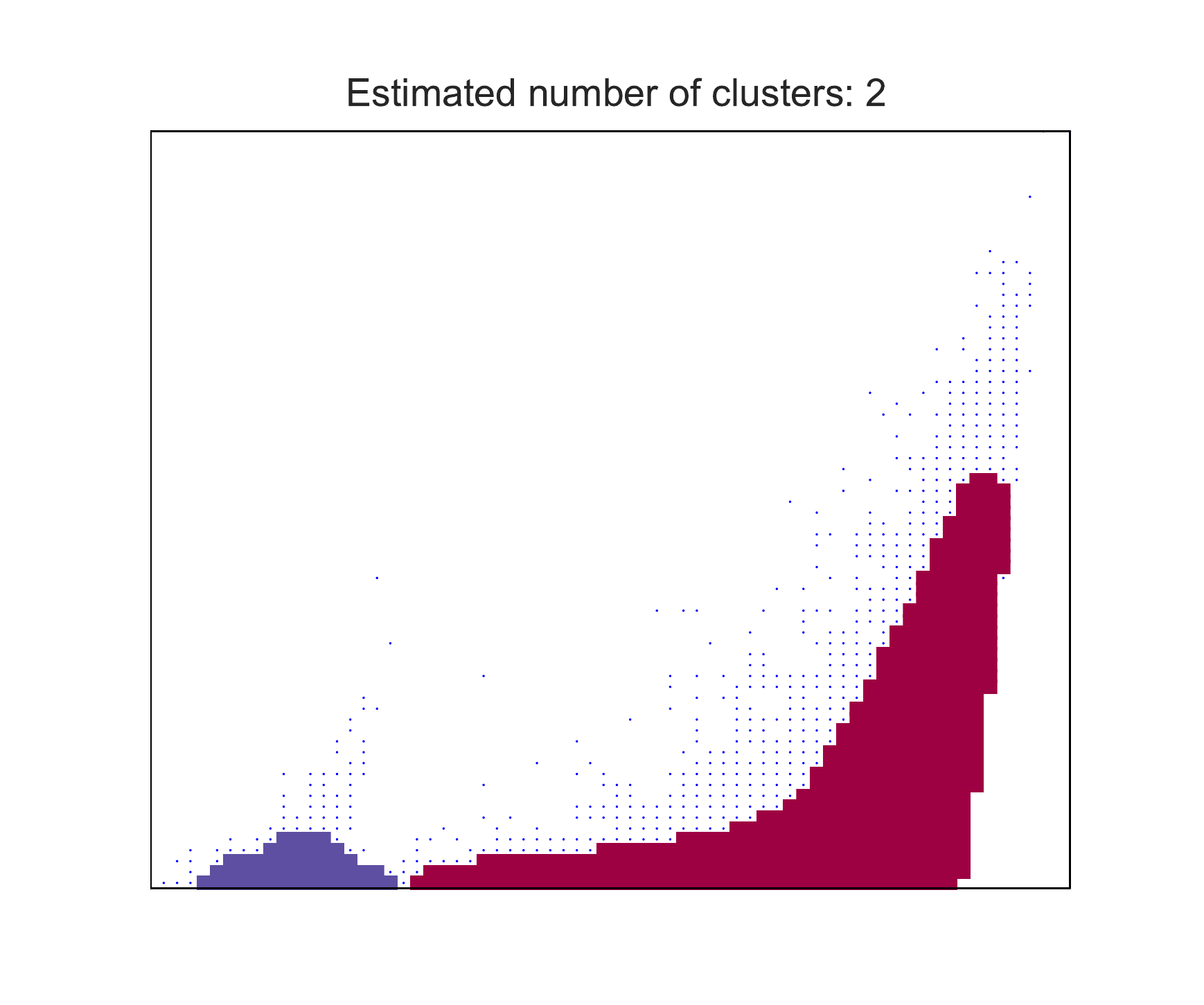}
		\caption{\centering{sting (\textbf{manual}) \{2\}}}
		\label{fig:amazon_outhub_dbscan}
	\end{subfigure}
	\begin{subfigure}[t]{0.11\linewidth}
		\centering
		\includegraphics[height=0.65in]{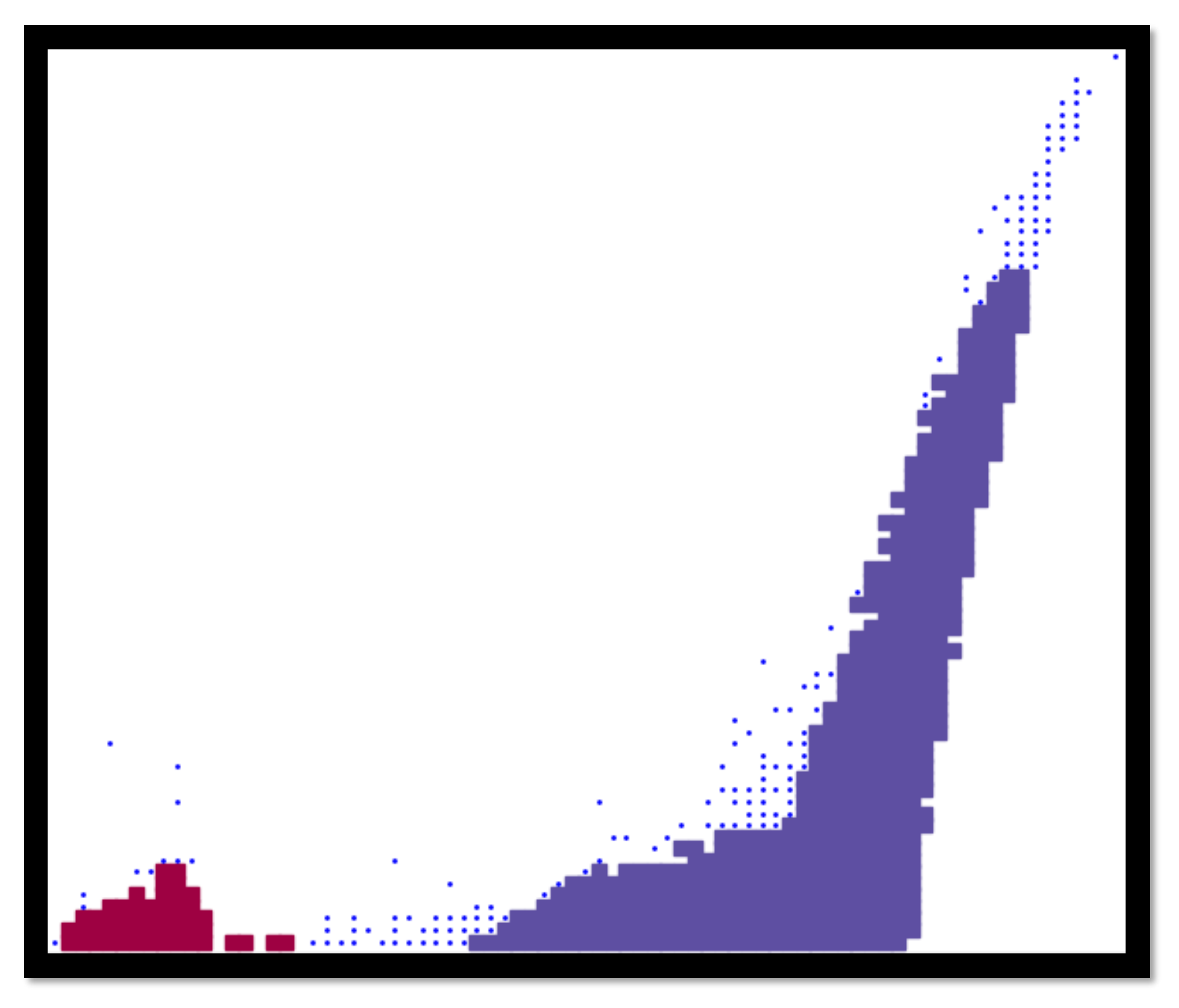}
		\caption{\centering{\eaglemine  \{2\}}}
		\label{fig:amazon_outhub_eaglemine}
	\end{subfigure}
	
	\begin{subfigure}[t]{0.145\linewidth}
		\centering
		\includegraphics[height=0.67in]{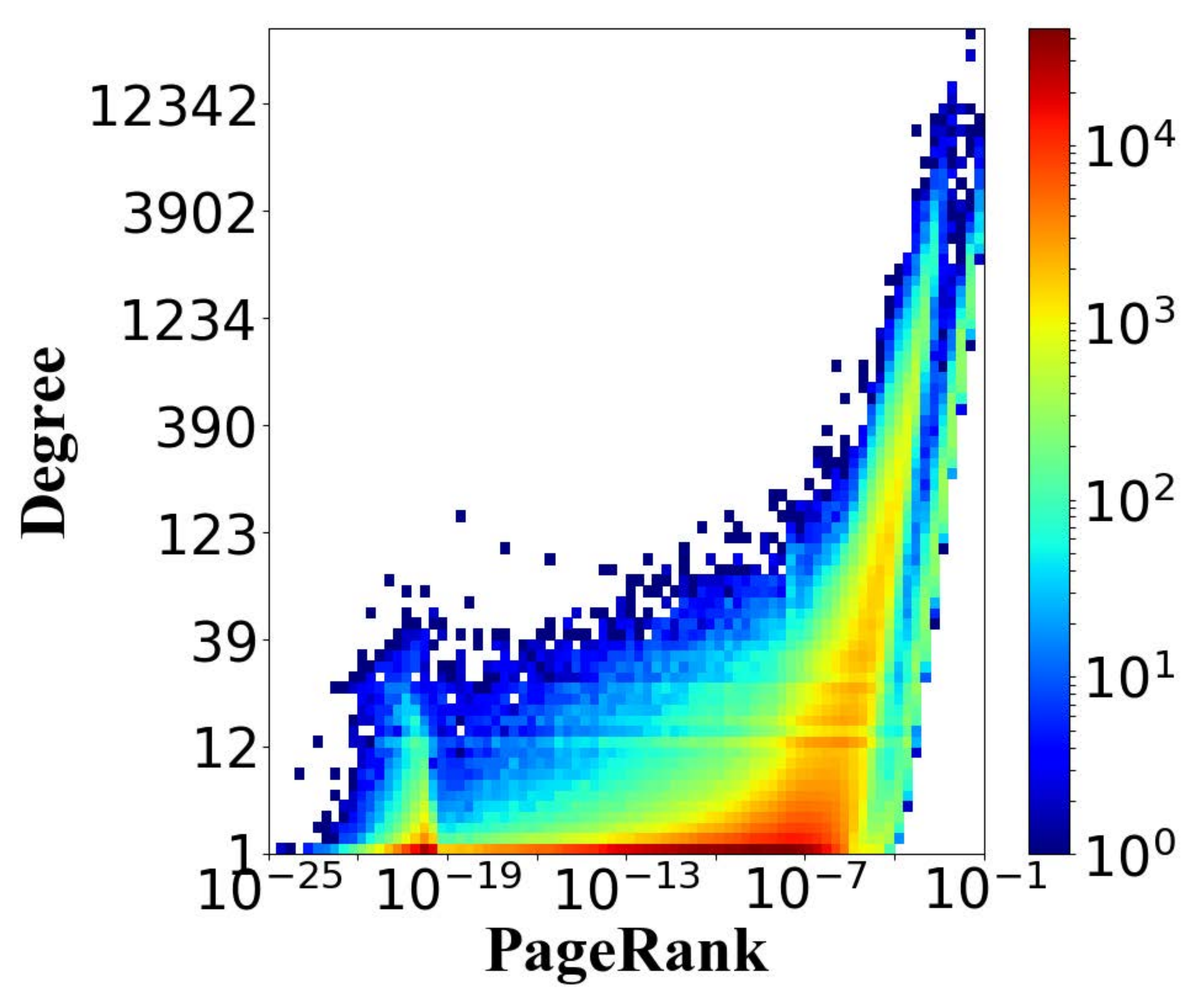}
		\caption{\centering{Degree vs. PageRank}}
		\label{fig:flickr_inaut}
	\end{subfigure}
	\begin{subfigure}[t]{0.11\linewidth}
		\centering
		\includegraphics[height=0.65in]{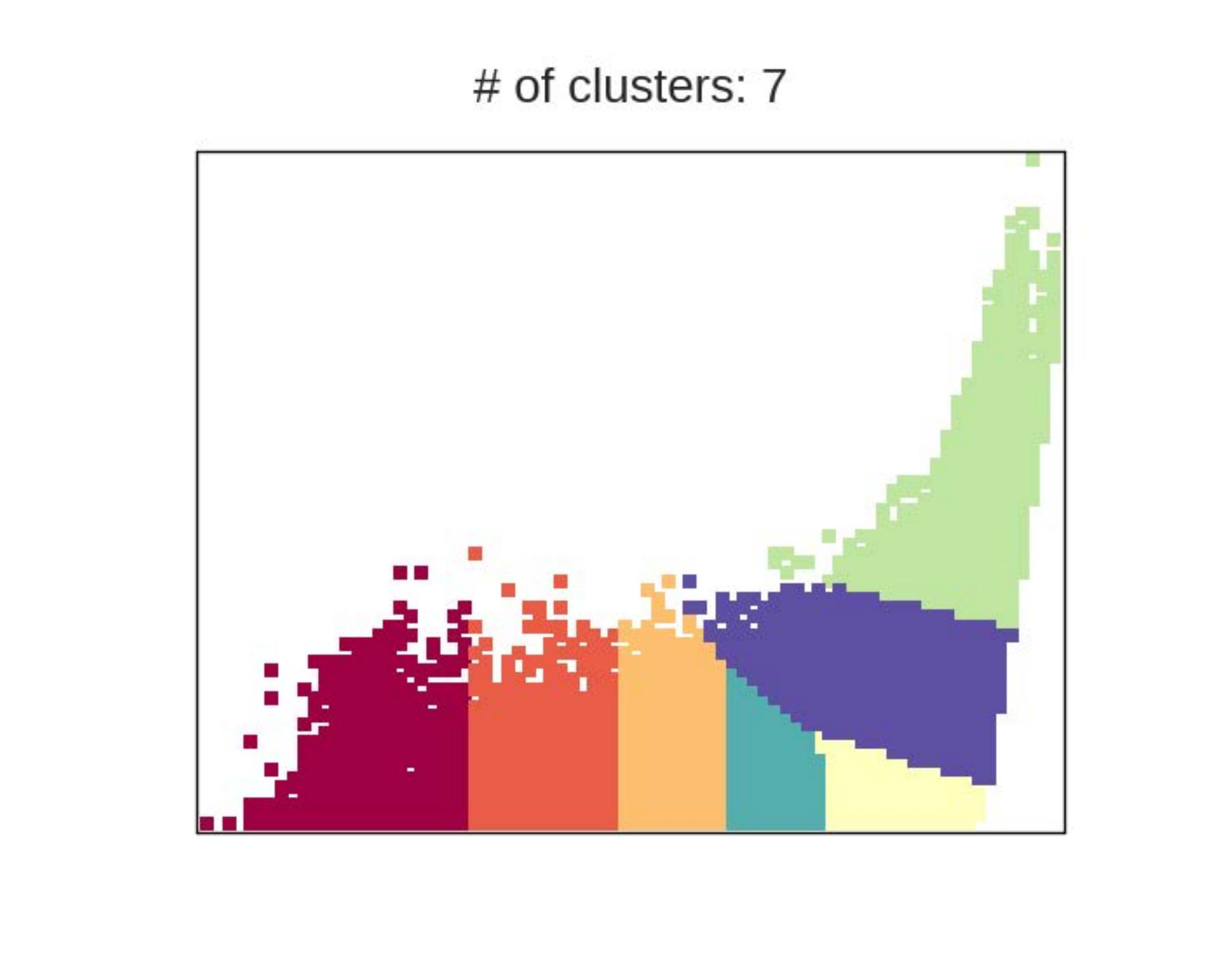}
		\caption{\centering{X-means \{5\}}}
		\label{fig:flickr_inaut_xmeans}
	\end{subfigure}
	\begin{subfigure}[t]{0.11\linewidth}
		\centering
		\includegraphics[height=0.65in]{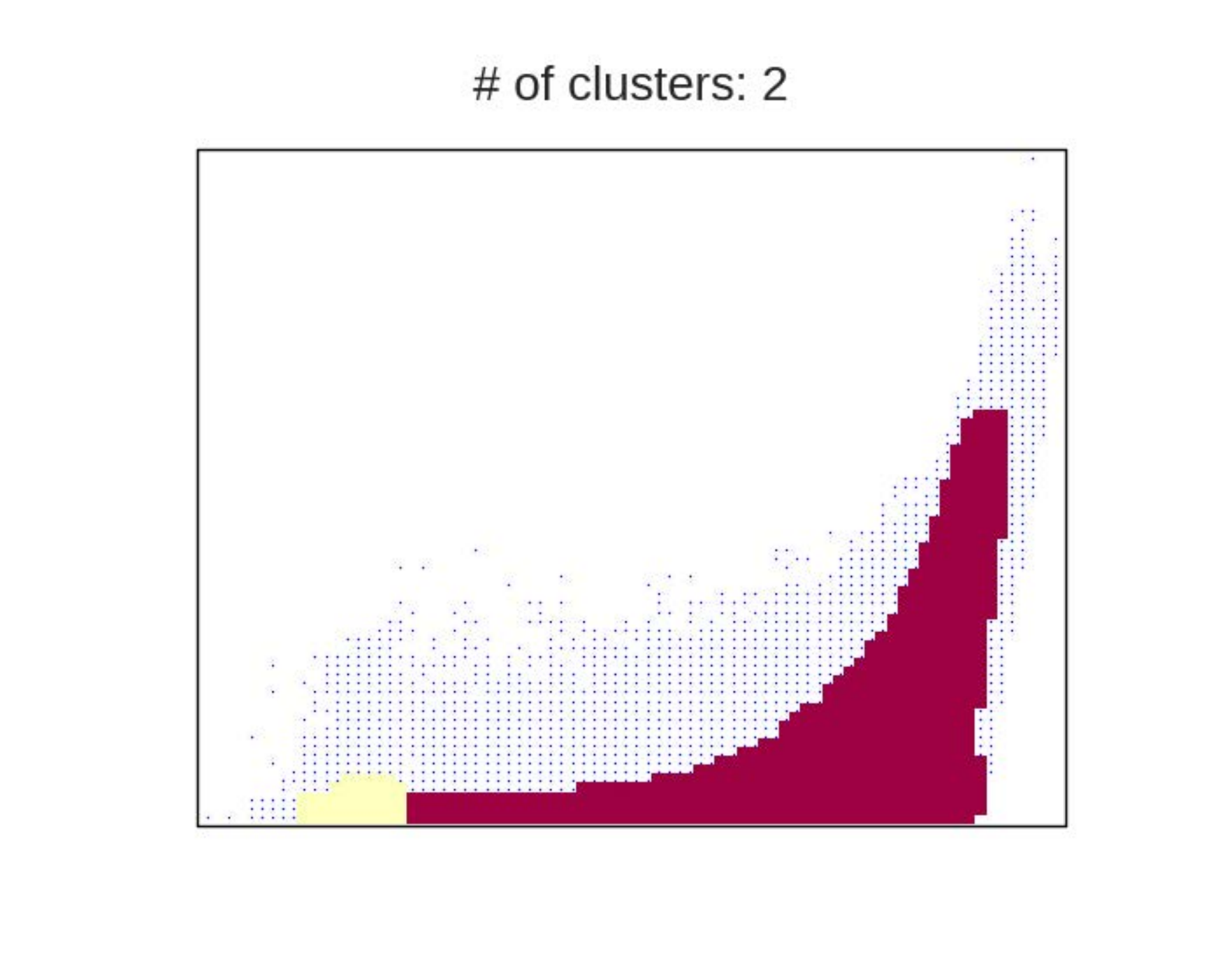}
		\caption{\centering{DBSCAN (\textbf{manual}) \{2\}}}
		\label{fig:flickr_inaut_dbscan}
	\end{subfigure}
	\begin{subfigure}[t]{0.11\linewidth}
		\centering
		\includegraphics[height=0.65in]{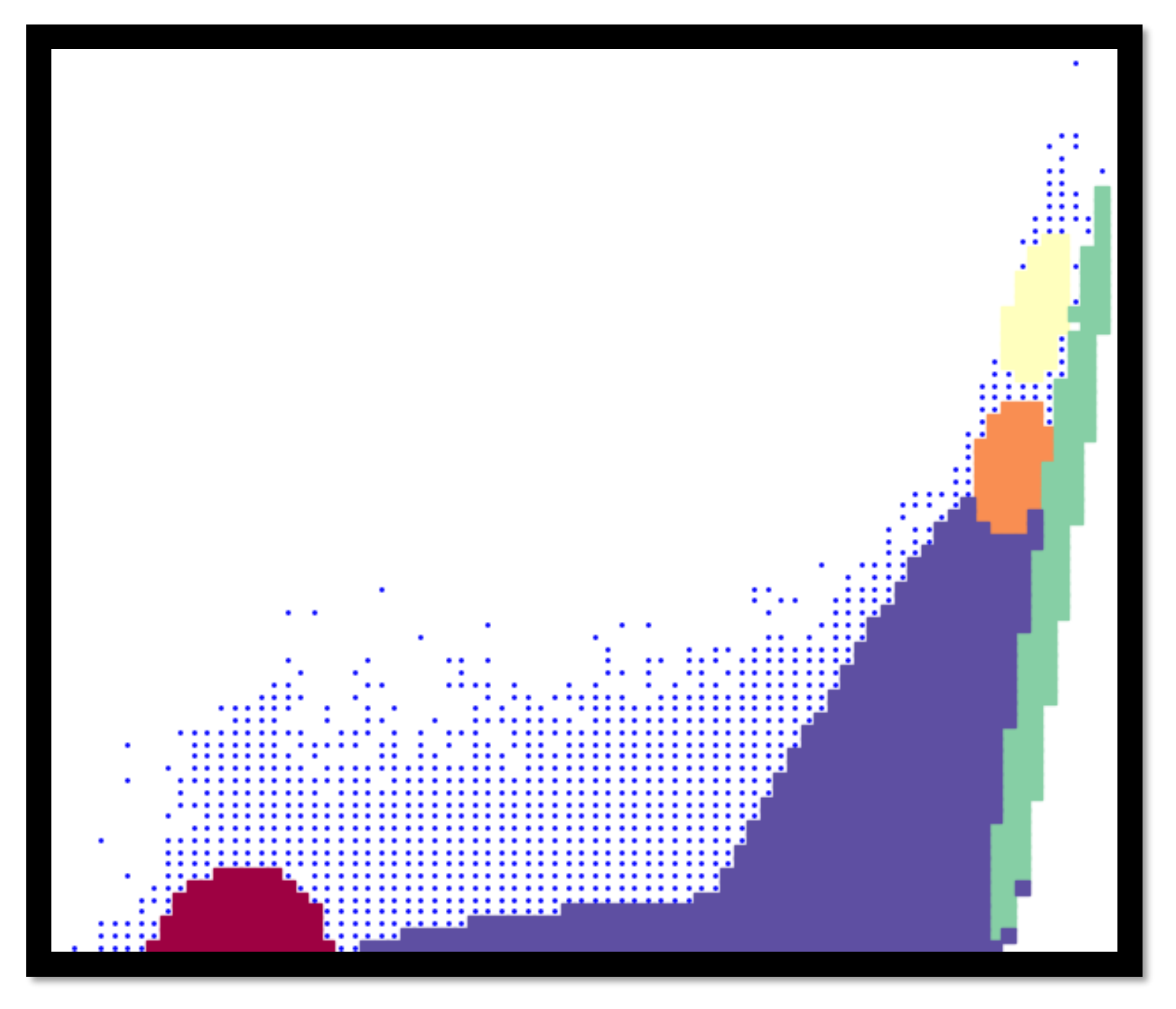}
		\caption{\centering{\eaglemine  \{5\}}}
		\label{fig:flickr_inaut_eagelmine}
	\end{subfigure}
	\vline
	\begin{subfigure}[t]{0.145\linewidth}
		\centering
		\includegraphics[height=0.67in]{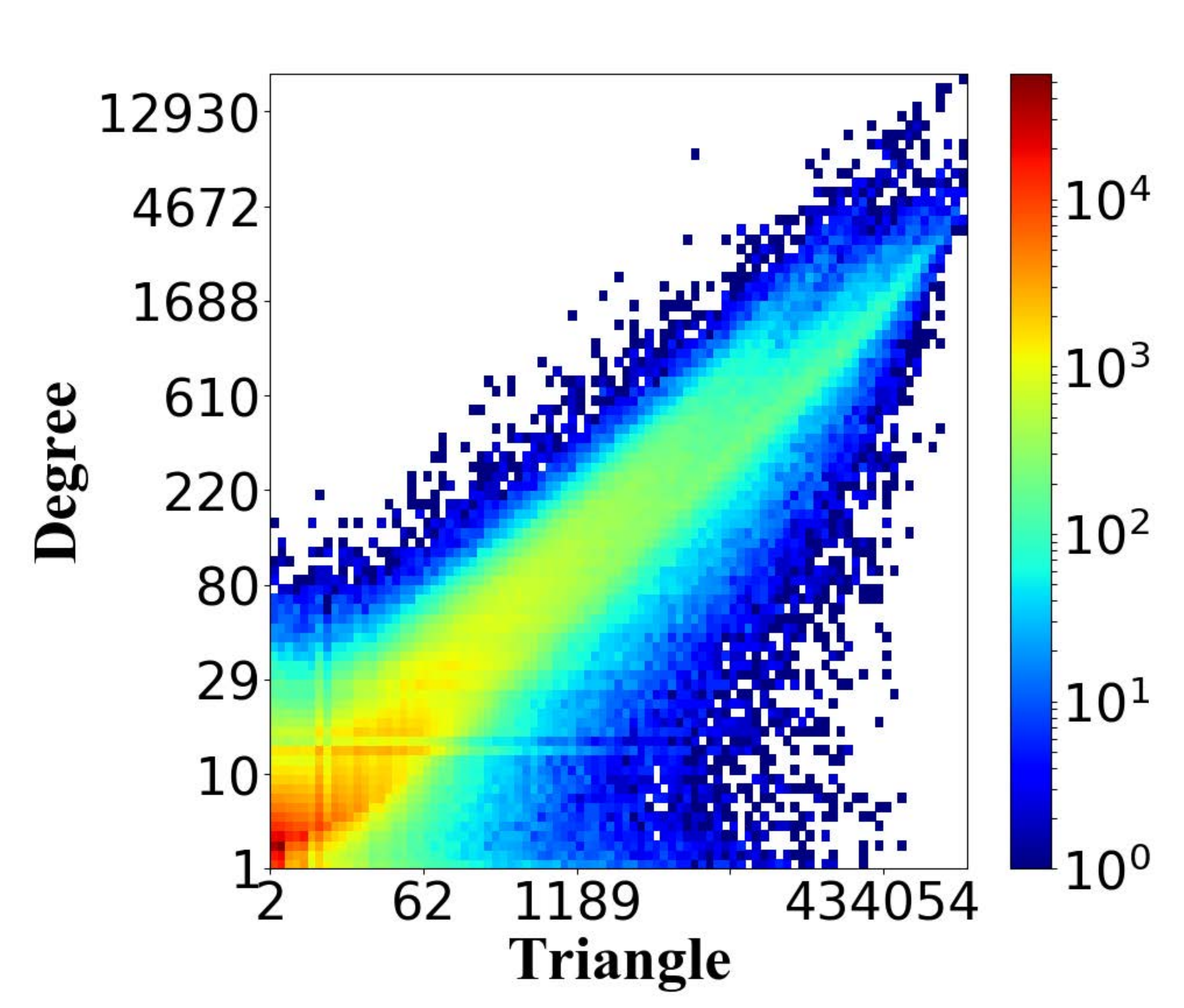}
		\caption{\centering{Degree vs. Triangle}}
		\label{fig:flickr_degtri}
	\end{subfigure}
	\begin{subfigure}[t]{0.11\linewidth}
		\centering
		\includegraphics[height=0.65in]{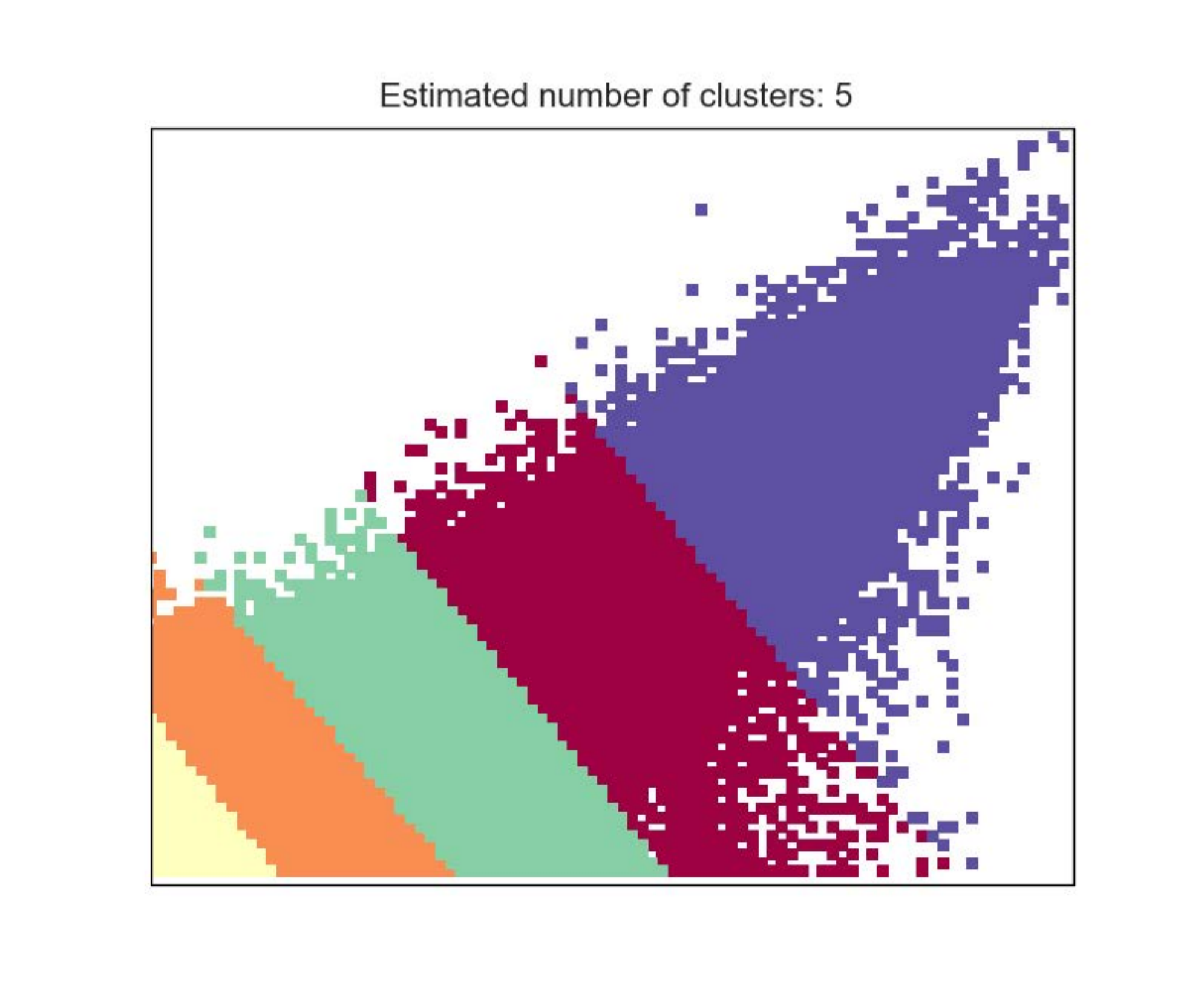}
		\caption{\centering{X-means \{5\}}}
		\label{fig:flickr_degtri_xmeans}
	\end{subfigure}
	\begin{subfigure}[t]{0.11\linewidth}
		\centering
		\includegraphics[height=0.65in]{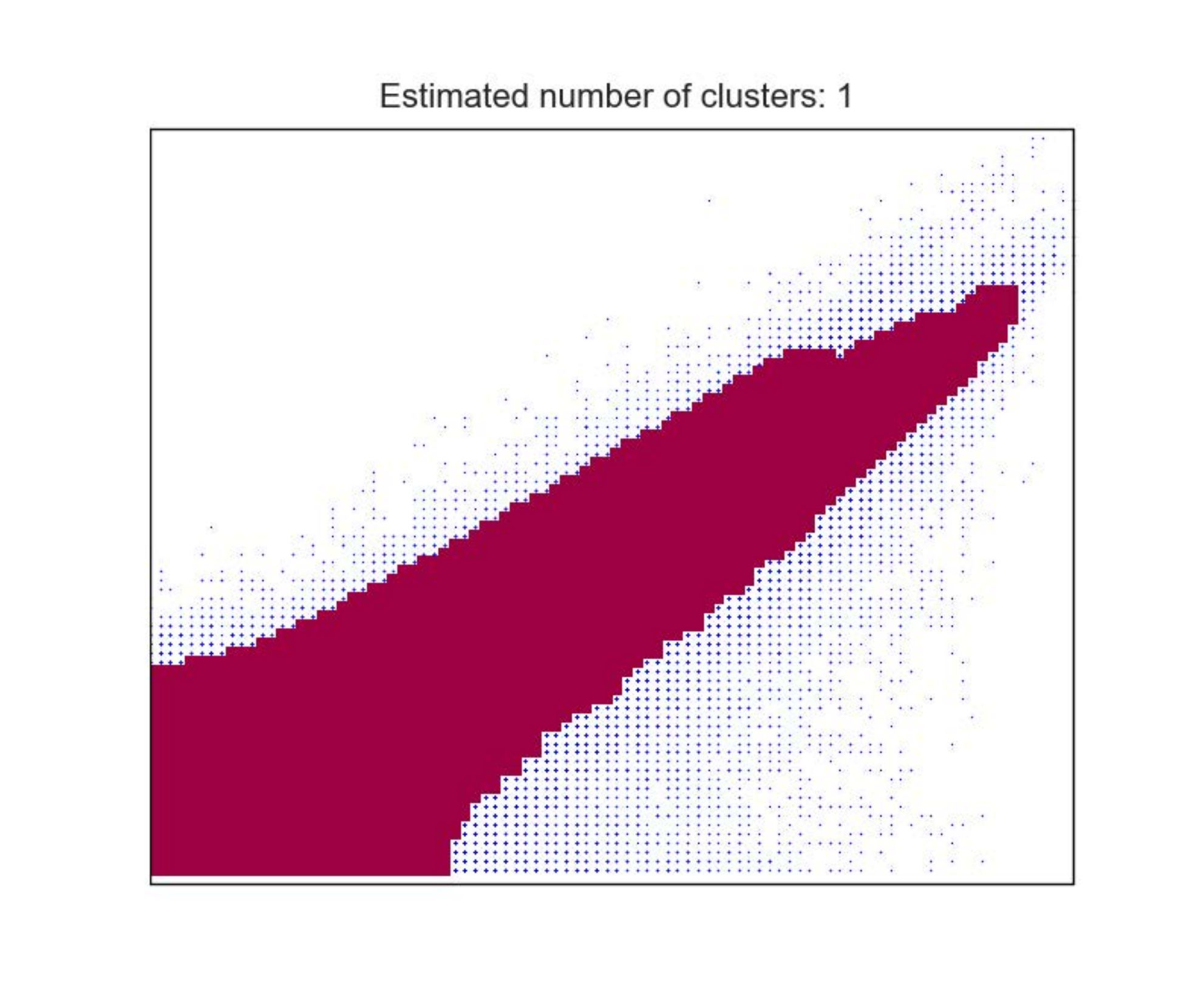}
		\caption{\centering{DBSCAN (\textbf{manual}) \{1\}}}
		\label{fig:flickr_degtri_dbscan}
	\end{subfigure}
	\begin{subfigure}[t]{0.11\linewidth}
		\centering
		\includegraphics[height=0.65in]{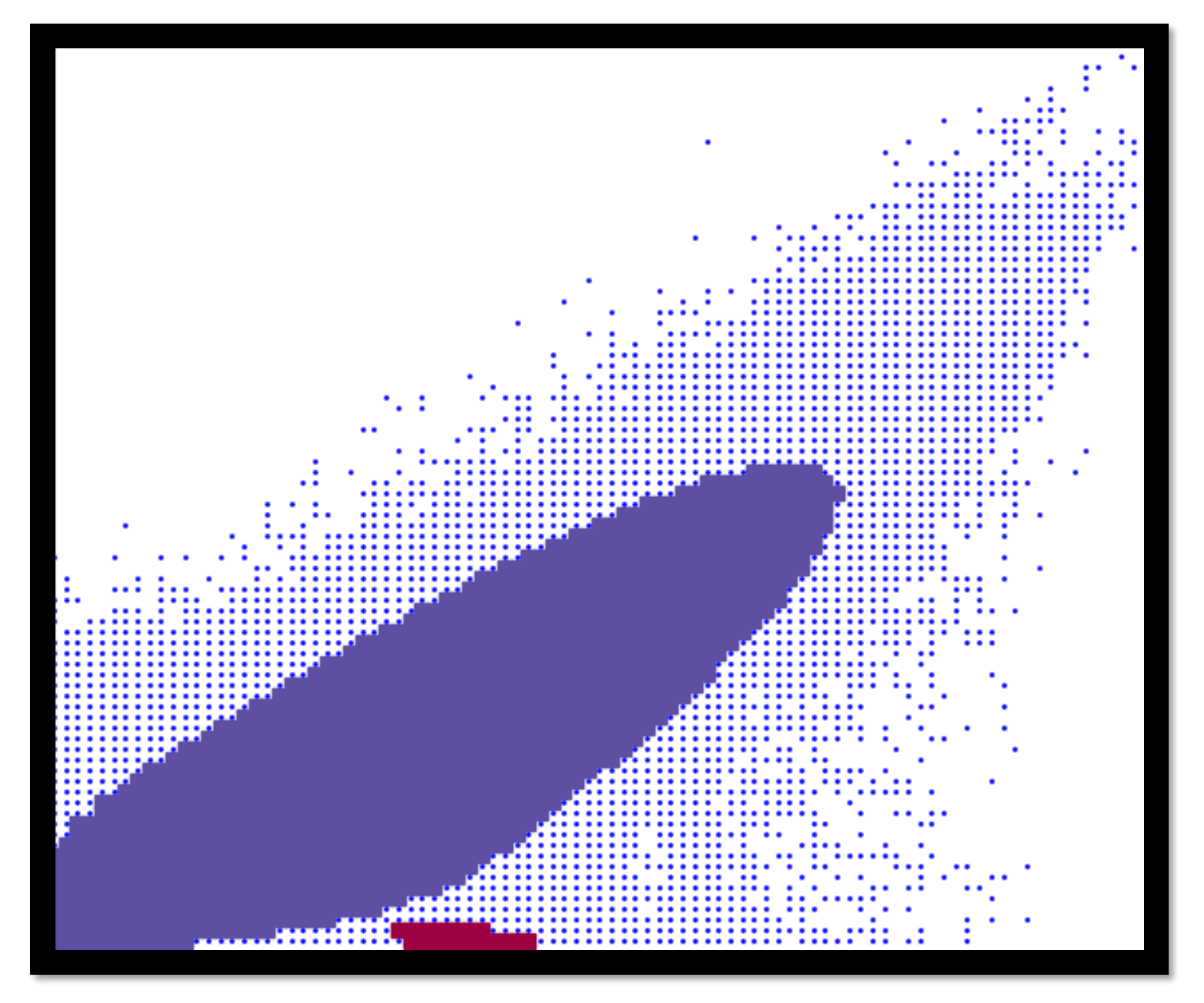}
		\caption{\centering{\eaglemine  \{2\}}}
		\label{fig:flickr_inaut_eagelmine}
	\end{subfigure}
	
	\begin{subfigure}[t]{0.145\linewidth}
		\centering
		\includegraphics[height=0.65in, width=1.0in]{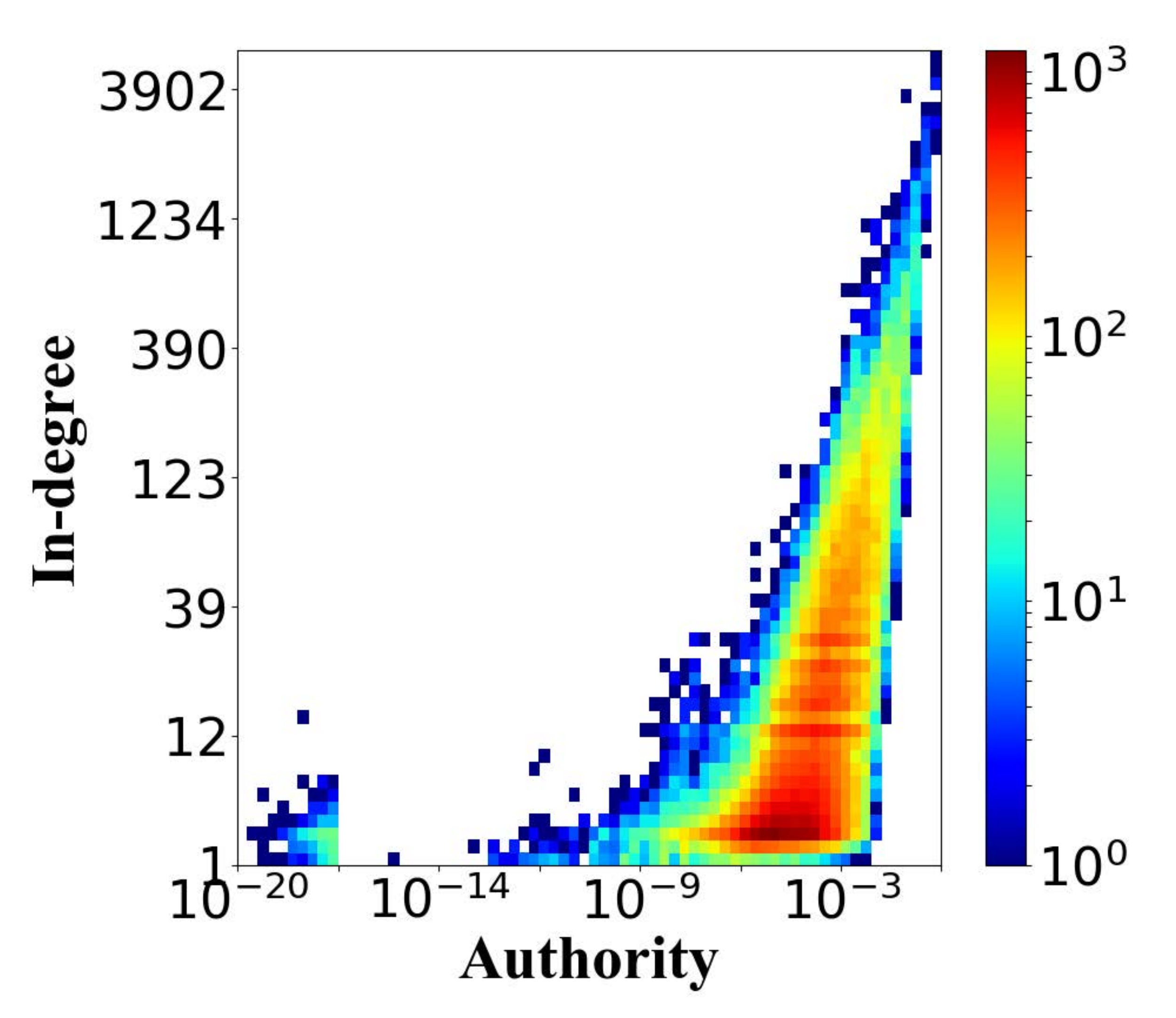}
		\caption{\centering{In-degree \,\, vs. Authority}}
		\label{fig:yelp_inaut}
	\end{subfigure}
	\begin{subfigure}[t]{0.11\linewidth}
		\centering
		\includegraphics[height=0.65in]{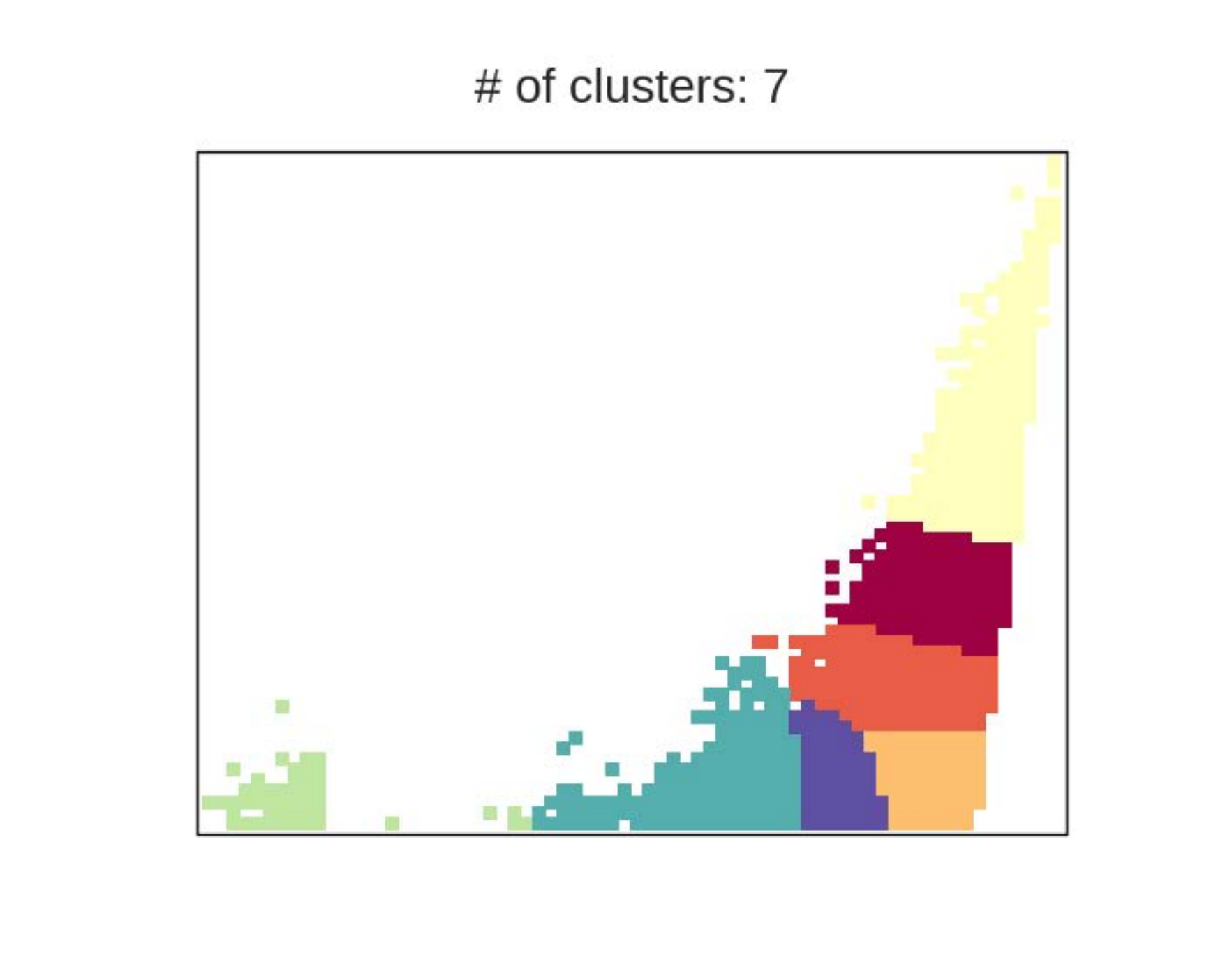}
		\caption{\centering{X-means \{7\}}}
		\label{fig:yelp_inaut_xmeans}
	\end{subfigure}
	\begin{subfigure}[t]{0.11\linewidth}
		\centering
		\includegraphics[height=0.65in]{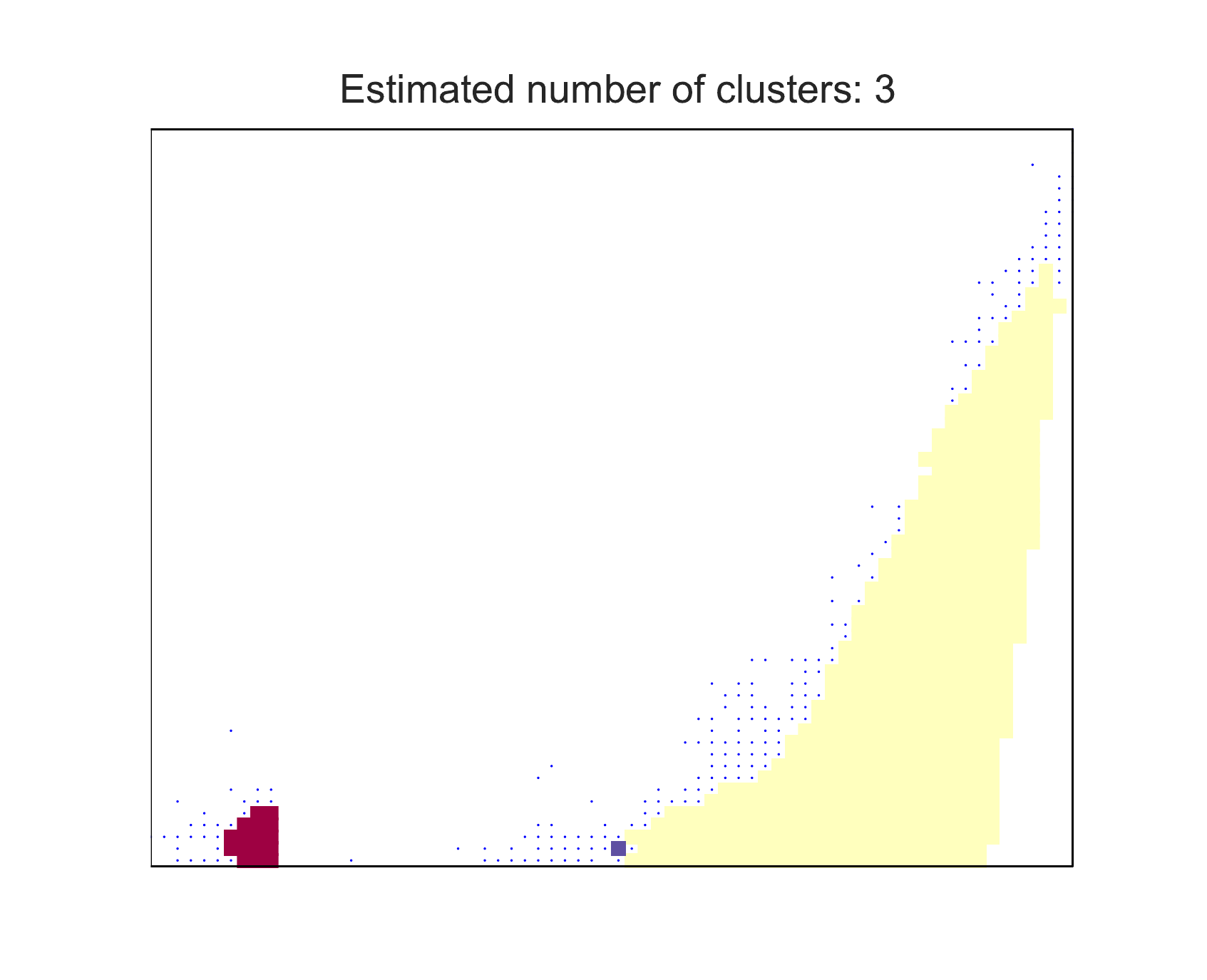}
		\caption{\centering{sting (\textbf{manual}) \{3\}}}        
		\label{fig:yelp_inaut_dbscan}
	\end{subfigure}\,
	\begin{subfigure}[t]{0.11\linewidth}
		\centering
		\includegraphics[height=0.65in]{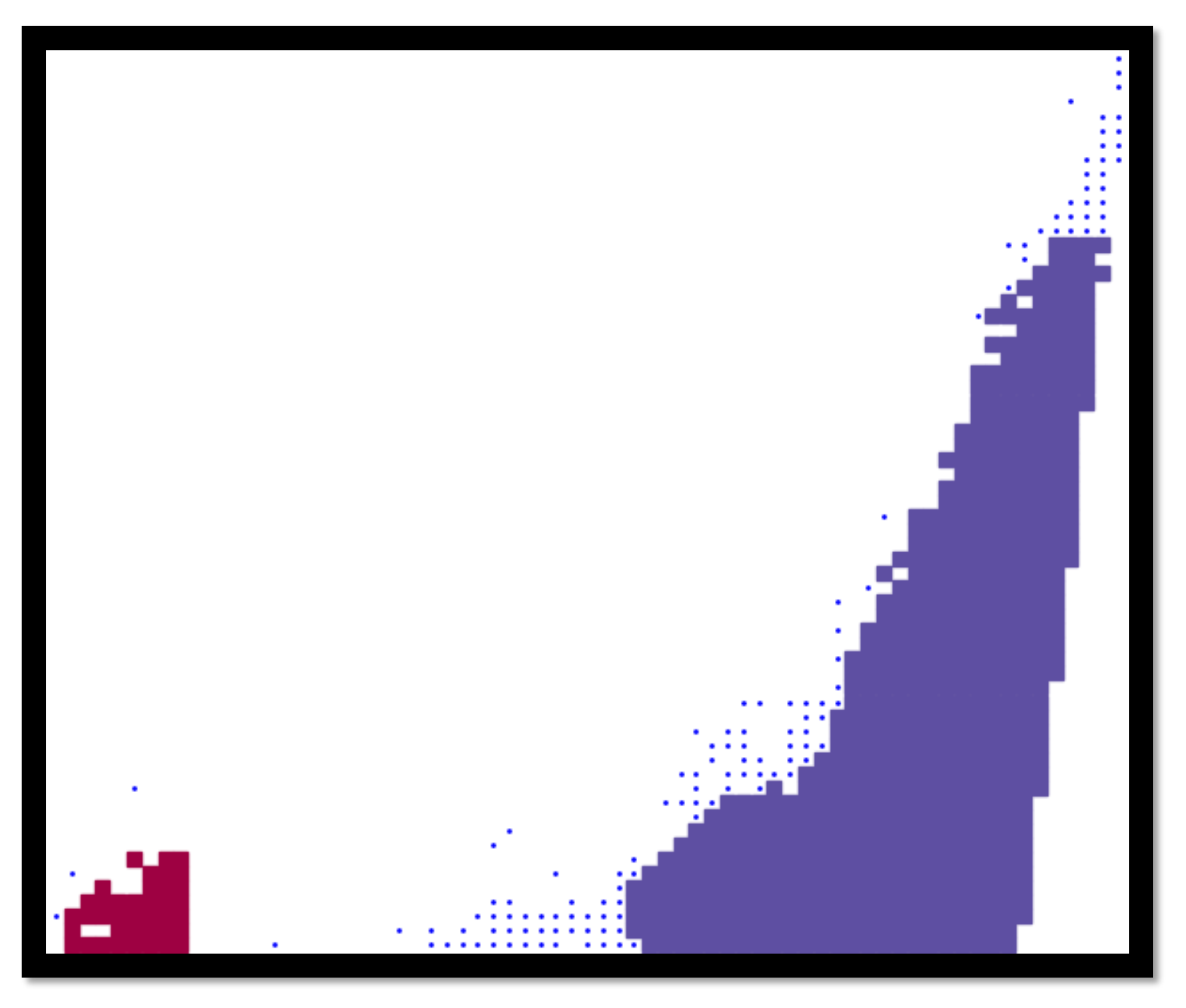}
		\caption{\centering{\eaglemine  \{2\}}}
		\label{fig:yelp_inaut_eaglemine}
	\end{subfigure}
	\vline 
	\begin{subfigure}[t]{0.145\linewidth}
		\centering
		\includegraphics[height=0.65in, width=1.0in]{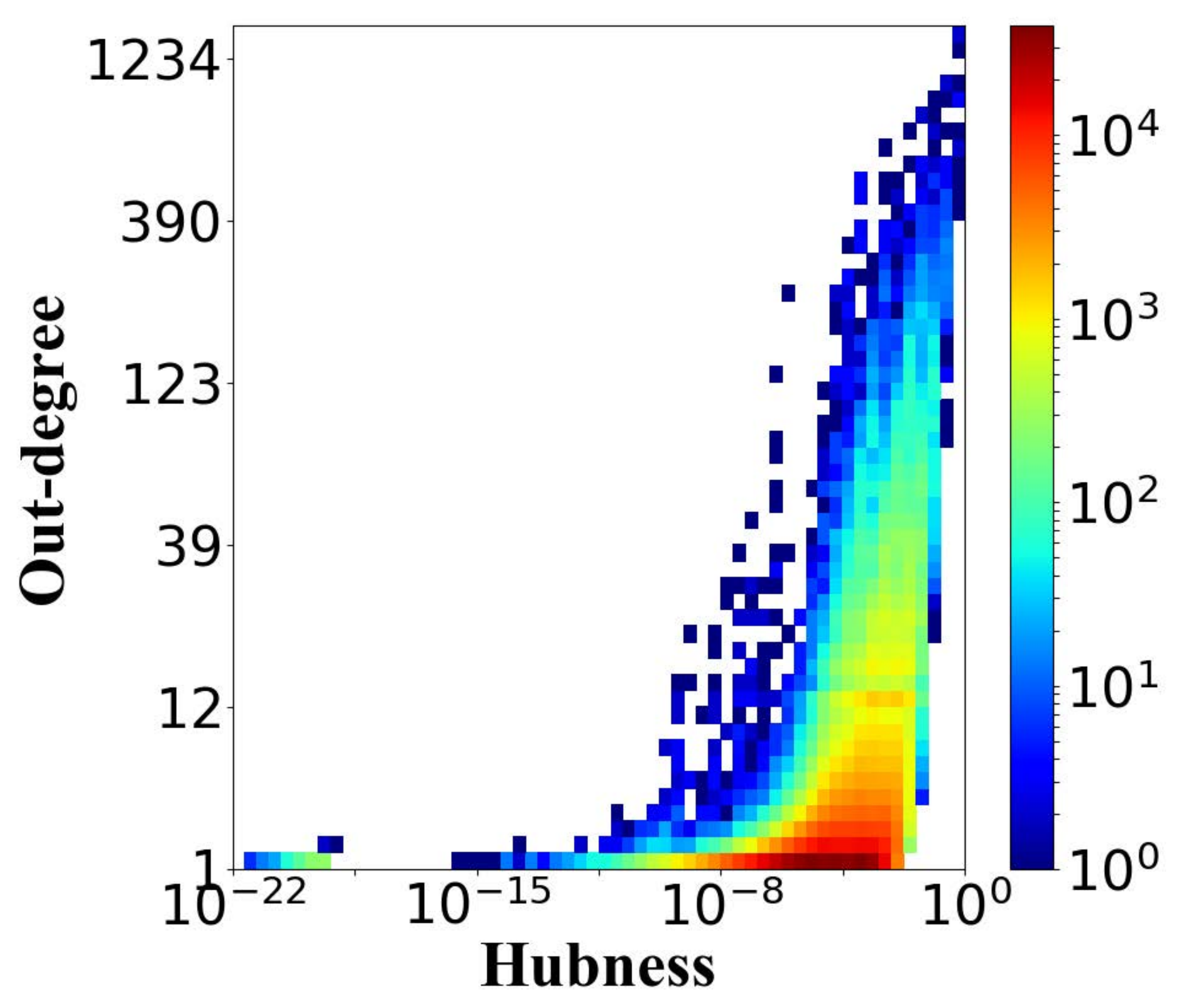}
		\caption{\centering{Out-degree \, vs. Hubness}}
		\label{fig:yelp_outhub}
	\end{subfigure}
	\begin{subfigure}[t]{0.11\linewidth}
		\centering
		\includegraphics[height=0.65in]{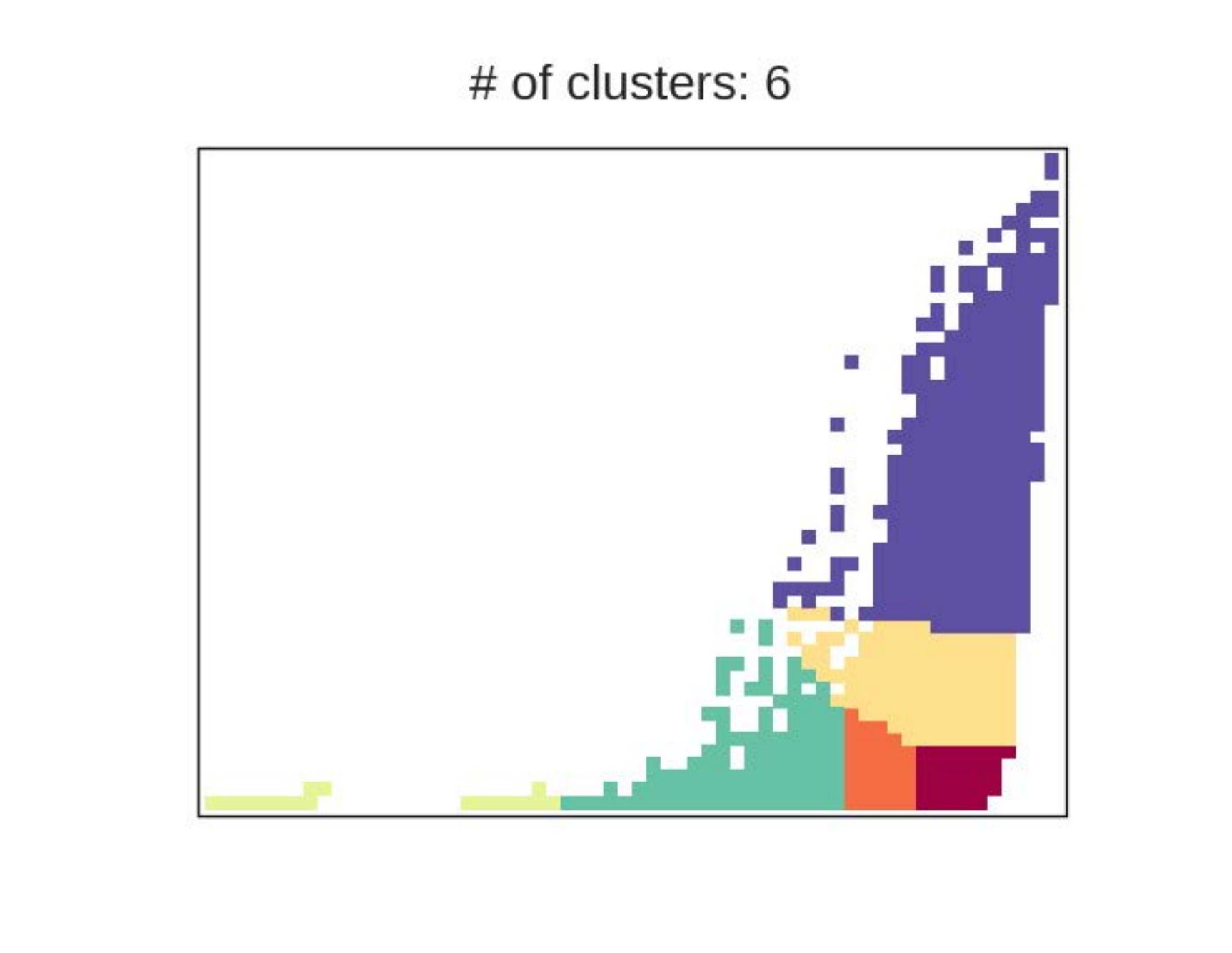}
		\caption{\centering{X-means \{6\}}}
		\label{fig:yelp_outhub_xmeans}
	\end{subfigure}
	\begin{subfigure}[t]{0.11\linewidth}
		\centering
		\includegraphics[height=0.65in]{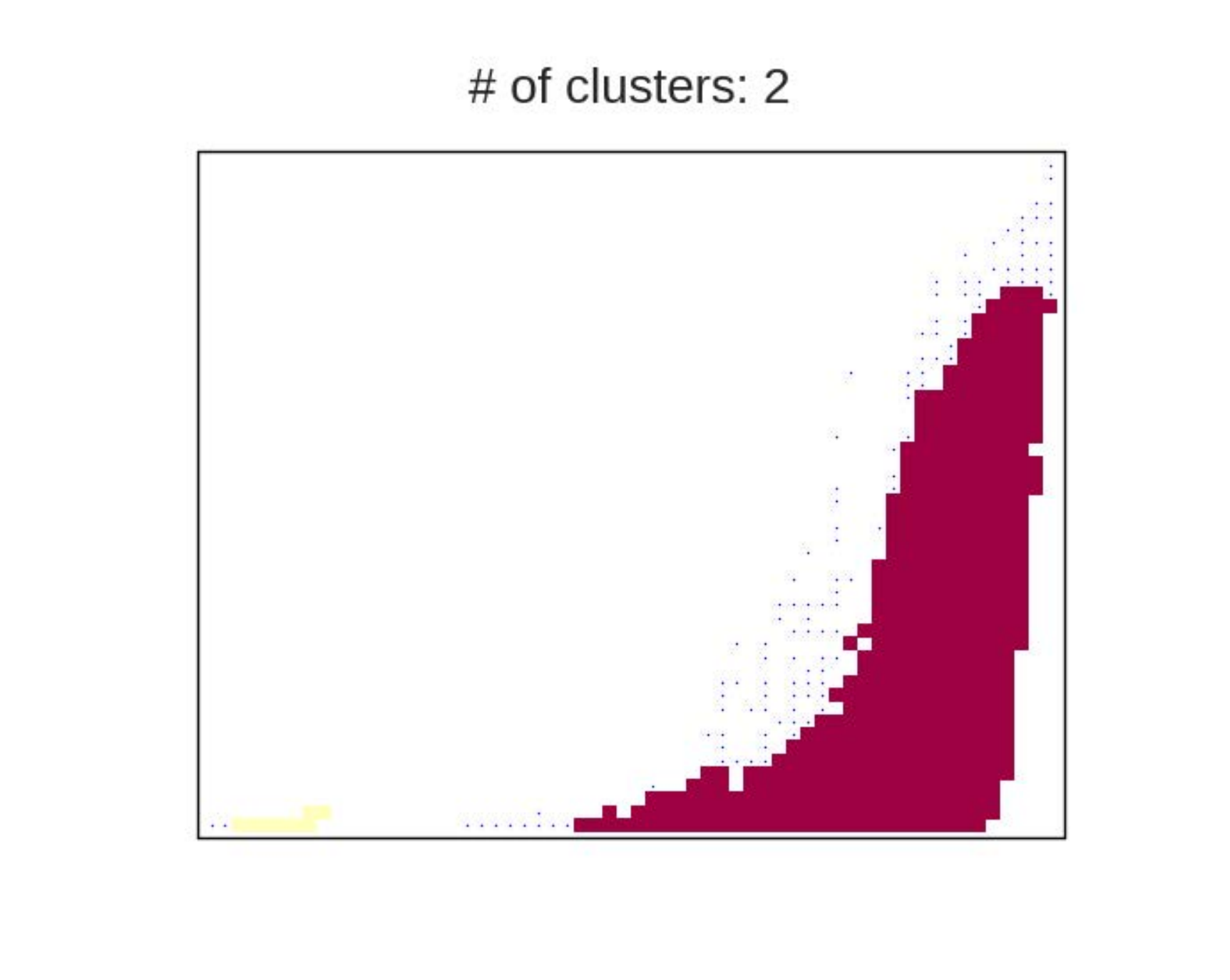}
		\caption{\centering{DBSCAN (\textbf{manual}) \{1\}}}
		\label{fig:yelp_outhub_dbscan}
	\end{subfigure}
	\begin{subfigure}[t]{0.11\linewidth}
		\centering
		\includegraphics[height=0.65in]{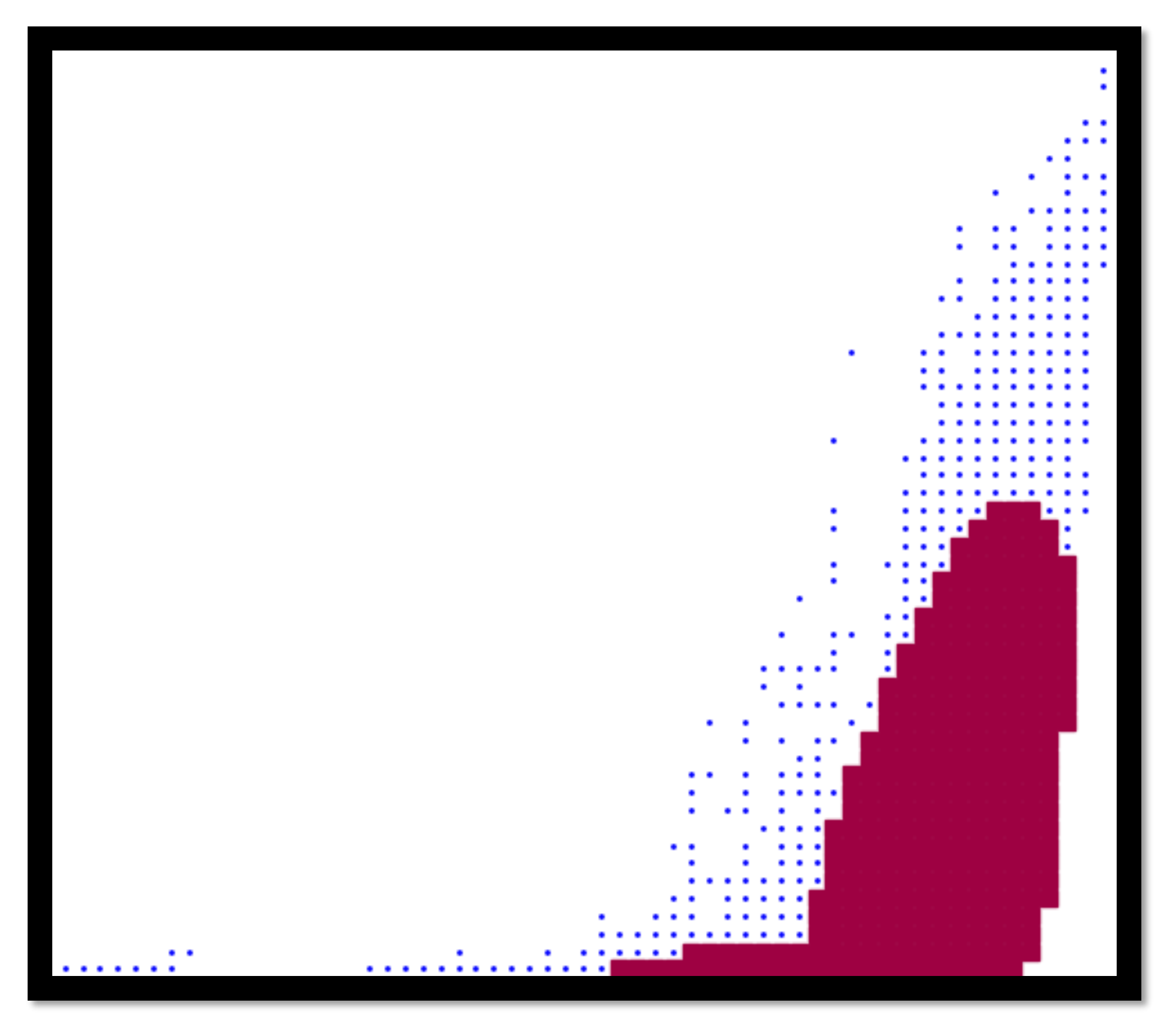}
		\caption{\centering{\eaglemine  \{1\}}}
		\label{fig:yelp_outhub_ealgemine}
	\end{subfigure}
	\vspace{-0.06in}
	\caption{{ Qualitative evaluation (vision-based)} on real-world datasets.
		\textbf{(a)-(d):} Sina weibo data.
		\textbf{(e)-(h):} Amazon movies data.
		\textbf{(i)-(p):} Flickr friendship data (Homogeneous graph);
		\textbf{(q)-(x):} Yelp products online review.
		The clusters found by \eaglemine agrees well with human vision.
		The number of clusters from each algorithm is in brackets 
		after the algorithm name.
	}
	\label{fig:vision_compare}
	\vspace{-0.21in}
\end{figure*}

We design experiments to answer following questions:
\begin{compactenum}[1.]
    \item \textbf{Quantitative evaluation on real data:}
          Does \eaglemine  give significant improvement in concisely 
          summarizing the graph? 
    \item \textbf{Qualitative evaluation (vision-based):} 
          Does \eaglemine accurately detect \mcls that agree with human vision?
    \item \textbf{Anomaly detection:} 
          Does \eaglemine detect suspicious patterns in real world graphs?
    \item \textbf{Scalability:} 
          Is \eaglemine scalable with regard to the data size?
\end{compactenum}

We test \eaglemine on a variety of real world datasets,
whose details are offered in Table \ref{table:exp_data}.

\begin{table}[h]
	\vspace{-0.1in}
	\setlength{\abovecaptionskip}{1pt} 
	\centering
	\caption{Graph datasets used in our experiments.}
	\resizebox{\columnwidth}{!}{%
		\begin{tabular}{ l | l  l  l }
			& \# of nodes &  \# of edges & Content \\ \hline
			Amazon rating\cite{konect:amazon-ratings}    & (2.14M, 1.23M) & 5.84M & Rate \\
			Android\cite{McAuley:2015}             		 & (1.32M, 61.27K) & 2.64M & Rate \\
			BeerAdvocate~\cite{mcauley2013amateurs}  	 & (33.37K, 65.91K) & 1.57M & Rate \\
			Flickr\cite{mcauley2012image}        		 & (2.30M, 2.30M) & 33.14M & Friendship \\
			Yelp\cite{yelp:yelp_data}                    & (686K, 85.54K) & 2.68M & Rate \\
			Youtube\cite{Mislove:2007}            		 & (3.22M, 3.22M) & 9.37M & Who-follow-who\\
			Sina weibo & (2.75M, 8.08M)           		 & 50.1M & User-retweet-msg \\
		\end{tabular}
	}
	\label{table:exp_data}
	\vspace{-0.22in}
\end{table}

\subsection{Q1. Quantitative Evaluation on Real Data}
\label{sect:quan_eval}
Envisioning the problem of clustering as a compression problem,
we use Minimum Description Length (MDL) as the metric to 
measure the conciseness of histogram summarization.
The best result $\mathcal{M}$ will minimize 
$L = L_{\mathcal{M}} + L_{\epsilon}$,
where $L_\mathcal{M}$ is the length in bits of 
description of the model, 
and $L_{\epsilon}$ is the description length of differences (errors)
between observed data and model expectation.
For the histogram $\mathcal H$, denote the model expectation as
$\tilde{h}$ for the number of node $h$ in grid $g$,
and with the description errors that can be decoded from
$L_\epsilon$, the original data can be accurately recovered as 
$h = \tilde{h} + \epsilon$.  
We extend Elias codes~\cite{elias1975universal,bohm2006robust} 
by encoding the sign with an extra bit, denoted as $\omega(\cdot)$.
Then we encode the total description error as
$L_\epsilon = \sum_{g} {\omega(h-\tilde h)}$.

Therefore, the total code length of \eaglemine is:
\begin{equation}
\setlength{\abovedisplayskip}{-0.01in}
\setlength{\belowdisplayskip}{-0.01in}
    \label{eqmdl}
    L = \omega(K) + \sum_i^{K} (5l_c + \omega(N_i)) + L_e
\end{equation}
where $l_c$ is the fixed code length for floating-point
number\footnote{$l_c$ = $32$ for floating-point numbers 
in our experimental setting.}
$K$ is the number of distributions used in the 
summary\footnote{$K$ = $2$ for mixture DTM Gaussian, otherwise $K = 1$.}
and each DTM Gaussian needs 5 free parameters. 
$N_i$ is the number of samples for distribution $i$.

We compared our summarization of \mcls with other
clustering algorithms:  X-means, G-means, 
DBSCAN\footnote{Since the performance of DBSCAN has 
already been manually tuned, we do not use OPTICS to 
search parameters for DBSCAN.},
and \eaglemine with digitized multivariate Gaussian distribution 
as vocabulary term, denoted as \eaglemine(DM).
In particular, G-means hierarchically split each cluster into
two based on hypothesis test of Gaussian distribution, 
so we viewed it as a summarization of clusters with the Gaussian
distribution. Hence, we first estimated the mean and covariance
for each cluster from G-means, 
and then calculated MDL as equation (\ref{eqmdl}).
As for other baselines, we calculated the MDL of clustering 
results as\cite{Chakrabarti:2004, bohm2006robust}.

The parameters of baselines were set as following
\begin{compactenum}[$\bullet$]
	\item X-means: initialized with $k$-means ($k=5$).
	\item G-means: $max\_depth=5$, which meant having 
	      16 clusters at most, to avoid too many fine-grained clusters.
	      P-value was $0.01$ which was not sensitive.
	\item DBSCAN: Eps=$1$, and used Manhattan distance. 
	      We searched MinPts from the mean of $h_i\in\mathcal H$ 
	      to the upper bound by a step of 50, and manually chose 
	      the best clusters.
\end{compactenum} 

The comparison of MDL is reported in Figure~\ref{fig:mdl_performance}.
We can see that \eaglemine achieves the shortest description length, 
indicating a concise histogram summarization.
On average, \eaglemine reduces the MDL code length more than 
$42.9\%, 36.4\%, 17.6\%$ compared with 
X-means, DBSCAN, and G-means respectively.
Our \eaglemine also outperforms \eaglemine (DM) without truncation.
Therefore, \eaglemine performs the best for histogram summarization 
with DTM Gaussian and overlapped distributions.

\begin{figure*}[t]
	\centering
	\begin{subfigure}[t]{0.3\linewidth}
		\centering
		\includegraphics[height=1.3in]{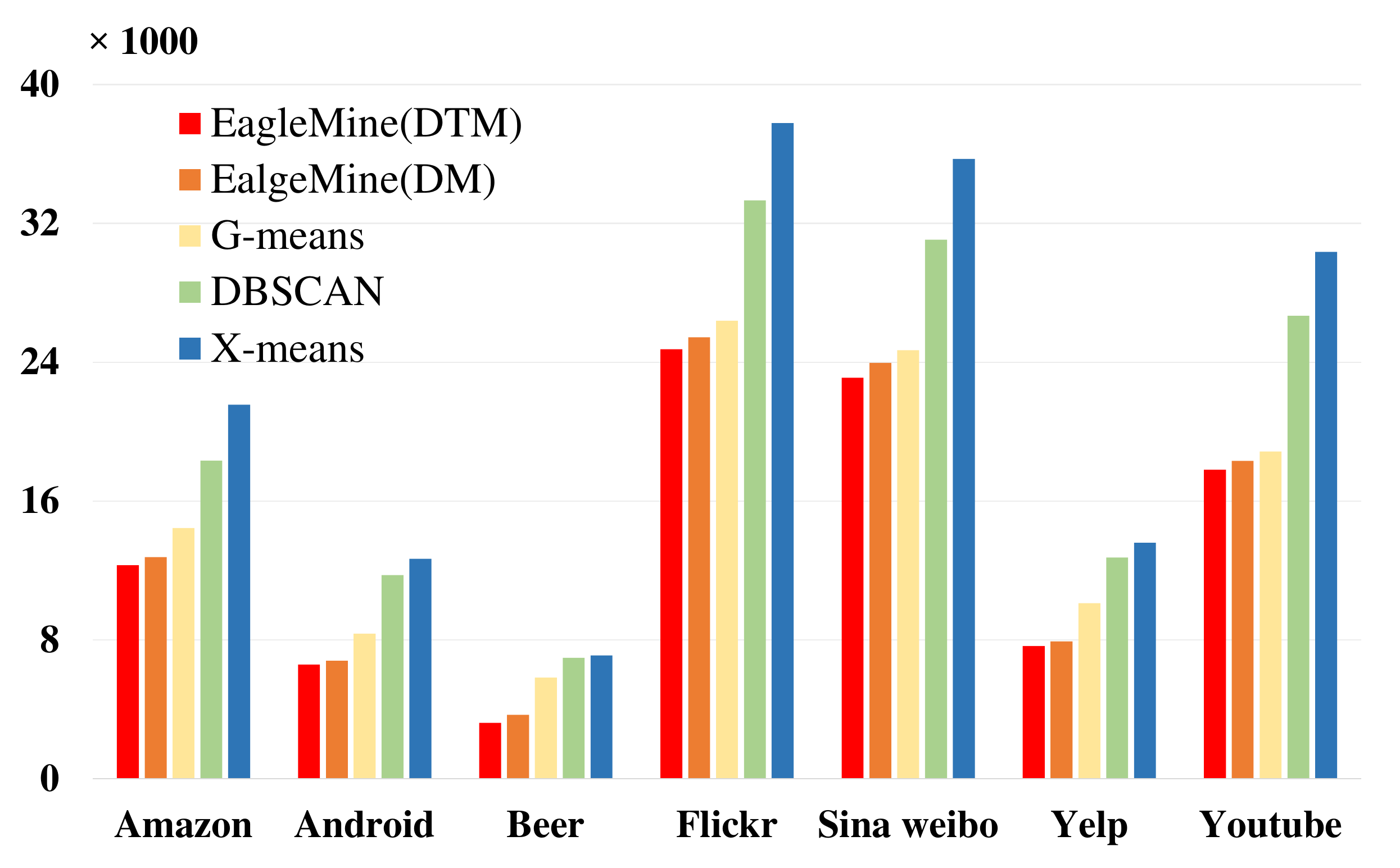}
		\caption{\centering MDL Quantitative Evaluation} 
		\label{fig:mdl_performance}
	\end{subfigure}
	\begin{subfigure}[t]{0.21\linewidth}
		\centering      	
		\includegraphics[height=1.3in]{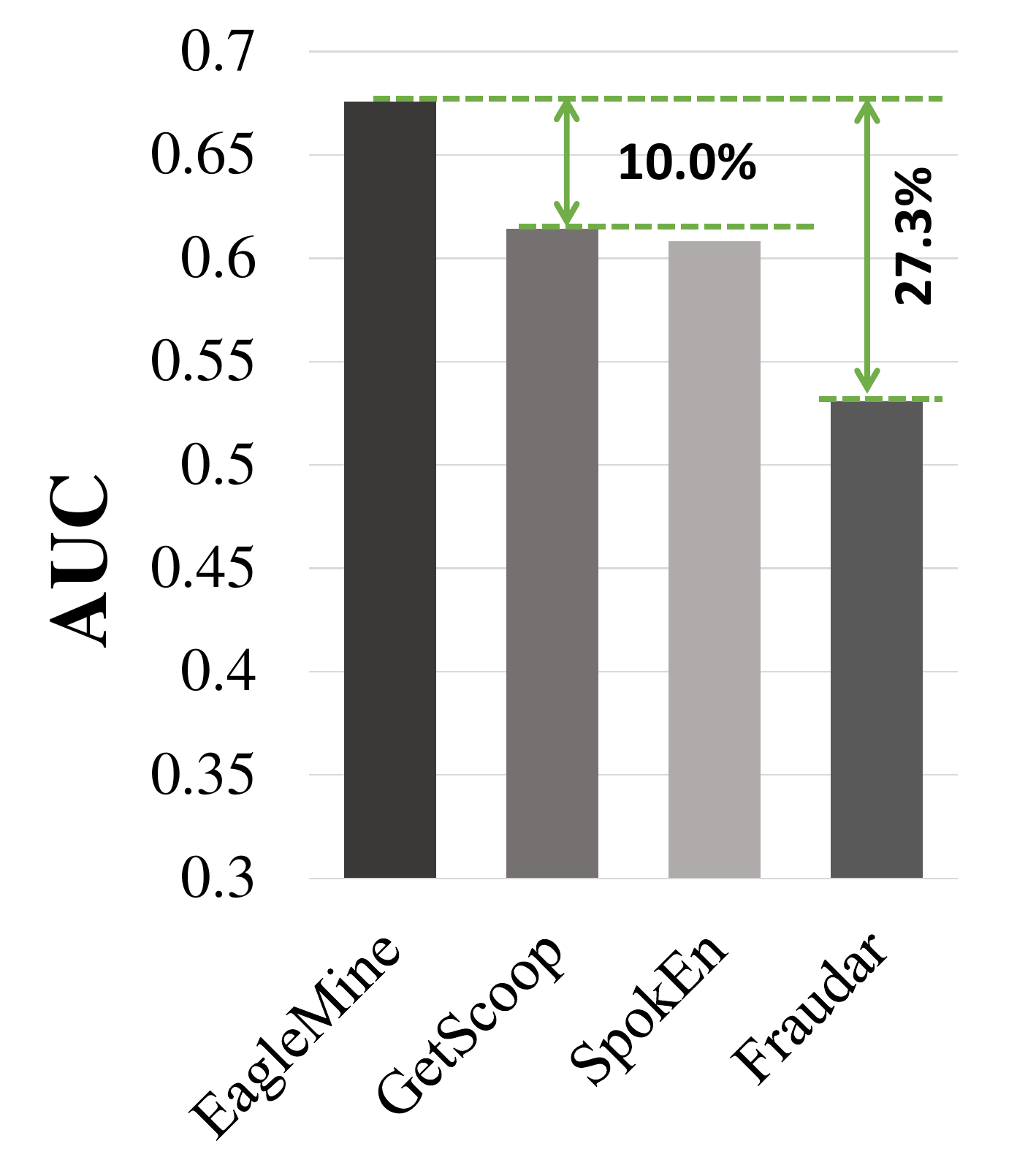}
		\caption{\centering AUC for suspicious users}
		\label{fig:auc_user}
	\end{subfigure}
	\begin{subfigure}[t]{0.23\linewidth}
		\centering   		
		\includegraphics[height=1.3in]{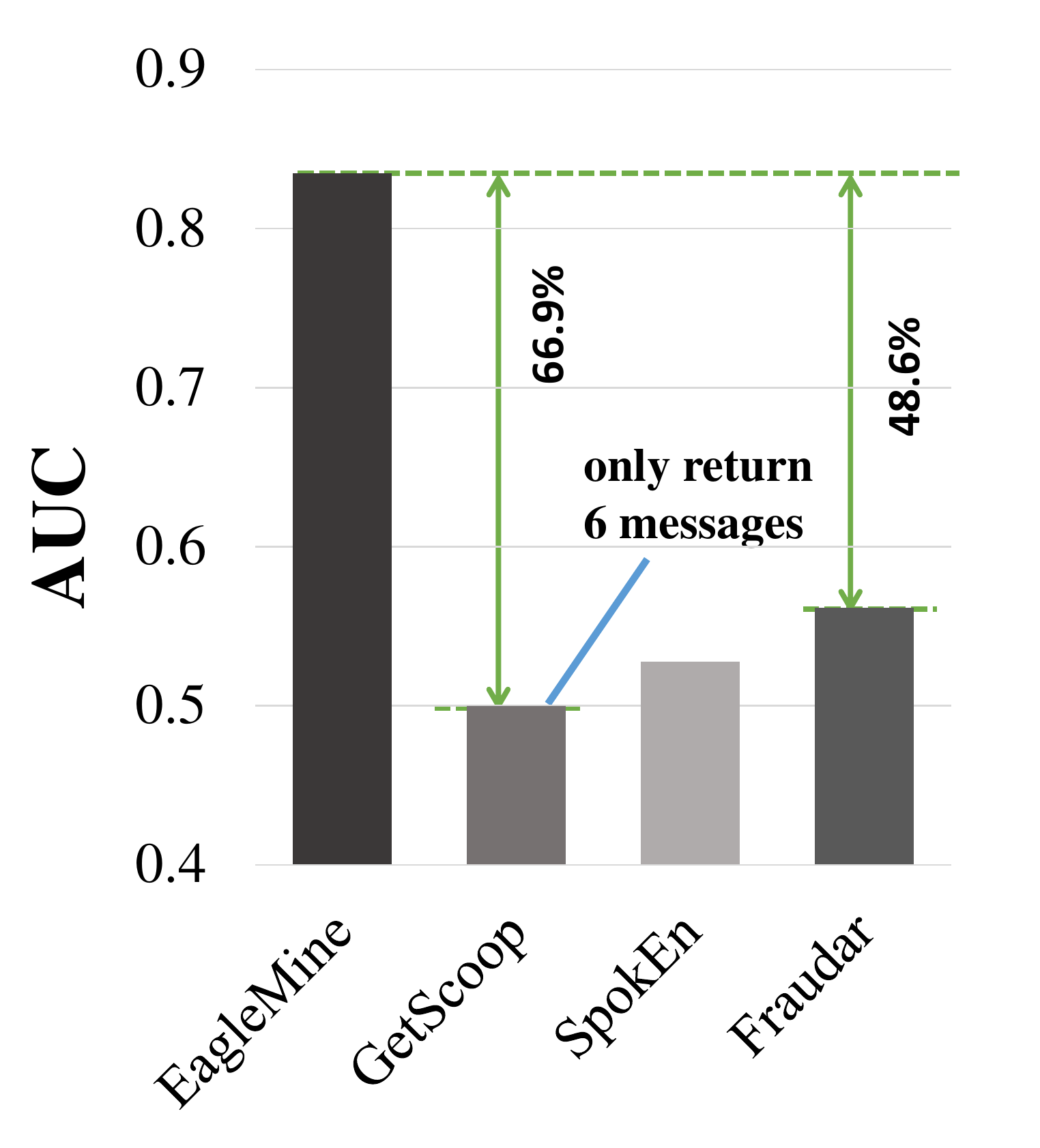}
		\caption{\centering AUC for suspicious messages}
		\label{fig:auc_msg}
	\end{subfigure}
	\begin{subfigure}[t]{0.21\linewidth}
		\centering
		\includegraphics[height=1.3in]{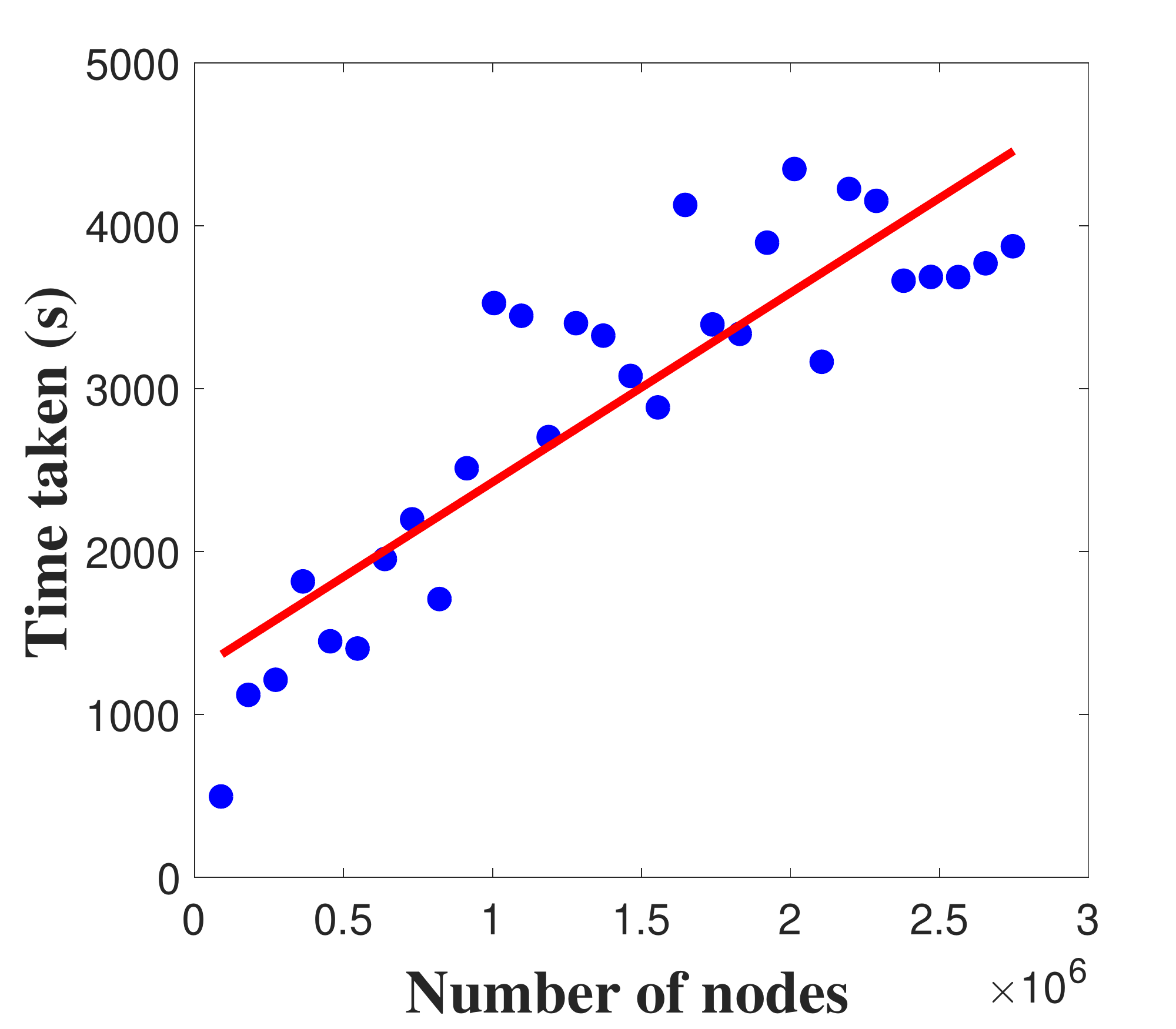}
		\caption{\centering  \eaglemine running time}
		\label{fig:run_time_fit}
	\end{subfigure}
	\vspace{-0.1in}
	\caption{
	    \textbf{\eaglemine Performance.}
		(a) MDL is compared on different real-world datasets. 
		    \eaglemine achieves the shortest description code length, 
			 which means concise summarization, and outperforms all 
		     other baselines.
		(b) and (c) \eaglemine has the best AUC for detecting suspicious 
		    users and messages on Sina weibo.
		(d) The curve \emph{blue} shows the running time of \eaglemine, 
		    compared to a linear function. }
	\label{fig:performace}
	\vspace{-0.24in}
\end{figure*}

\subsection{Q2. Qualitative evaluation (vision-based)}
\label{sect:qual_eval}

In this part, we visually compare the clustering results from all methods.
Based on the summarization $\pmb C$ of \Eaglemine,
the grids can be assigned to the distribution in $\pmb C$ and
labeled as $\pmb A$. With $\pmb A$ and outliers,
\eaglemine can be compared with other clustering algorithms.

Figure\ref{fig:vision_compare} exhibits 
the results of X-means, DBSCAN, STING, and \eaglemine
on some histograms of out-degree vs. hubness, 
in-degree vs. authoritativeness, 
degree vs. page-rank, and degree vs. triangle.
Excluding outliers labeled with \emph{blue} color on plots,
the number of clusters found by each method is listed in the bracket 
after corresponding algorithm name.
From the figures, it will be naturally expected that grids with similar
color (density) and near location should be as one cluster, and denser 
parts should be covered as many as possible.
The results show that X-means produces a number of poor quality clusters, 
as it divides the triangle-like island into pieces, 
which should be as one cluster, into different spherical b
locks in a way inconsistent with vision-based judgment.

Although manually tuned DBSCAN captures dense regions,
it overlooks some important micro-clusters as
Figure\ref{fig:wbeaglemine} shown.
Moreover, the performance of DBSCAN relies on extensive manual work to 
pick a good result.
Our \eaglemine gives the best results for mining clusters,
which is better than density-based clustering methods,
and also captures some small micro-clusters and outliers,
coinciding with human visual intuition.
Moreover, \eaglemine provides cluster descriptions using distributions,
which summarizes data with statistical characteristics.
Of course, \eaglemine does not yield these micro-clusters or outliers
if the features under-study do not show clustering structure 
as in Figure\ref{fig:yelp_outhub_ealgemine}.

The STING algorithm is equivalent to DBSCAN in an 
extreme condition\cite{wang1997sting}.
It hence achieves very similar result to DBSCAN (see Figure\ref{fig:wbsting}).
Besides, it does not separate the three micro-clusters and the majority
cluster in red color.
Therefore, our \eaglemine gives better results for clustering 
and has more explainable results as manual recognition.
It is worth noticing that some small micro-clusters captured by 
\eaglemine are interesting patterns or suspicious groups.

\subsection{Q3. Anomaly detection}
In this section, we analyze the anomaly patterns in \mcls
of Sina Weibo dataset, which is a user-retweet-message graph.
The data was collected through the whole month of Nov 2013.
The statistics information 
is listed in Table~\ref{table:exp_data}.

Surprisingly, we found a cluster of users who had suspicious name
patterns. As Figure\ref{fig:sina_outdhub_cases} shows,
these accounts share several prefixes, 
e.g. 182 usernames with prefix ``best'',
and 178 usernames with prefix ``black''.
Moreover, since those users gathered in a cluster with small
covariance, they had similar out-degree 
$123 \pm 5$ (i.e.~number of retweets),
and similar hubness, indicating that they retweeted the same messages,
or messages with equal popularity.
Based on the analysis, these users are possibly bots with synchronized
behaviors\footnote{Since it is hard to crawl all the data from Sina Weibo,
we cannot verify this conjecture. But those users are 
still active in the online system.}.

To compare the performance for anomaly detection,
we labeled these nodes, both users and messages, 
from the results of baselines, and sample nodes of our 
suspicious clusters of \eaglemine 
considering it is impossible to label all the nodes.
Our labels were based on the following rules\cite{hooi2016fraudar}:
\textit{1} user-accounts/messages which were deleted from the online 
system (system operators found the suspiciousness)
\footnote{The status is checked in May 2017 with API of Sina Weibo.};
\textit{2} users with specific login-names pattern and other suspicious signals 
e.g.~approximately the same sign-up time, friends and followers count as discussed above;
\textit{3} messages talking about advertisement or retweeting promotion.

In total, we labeled 5,474 suspicious users and 4,890 suspicious messages.
We compared with state-of-the-art fraud detection algorithms 
GetScoop\cite{jiang2014inferring}, \textsc{SpokEn}\cite{prakash2010eigenspokes},
and Fraudar\cite{hooi2016fraudar}, 
and the results are reported in Figure~\ref{fig:auc_user}~\ref{fig:auc_msg}.
We sorted our suspicious nodes in the descendant order of 
hubness or authority scores. 
The AUC (the area under the ROC curve) was used as the metric.
The results show that \eaglemine achieves more than $10\%$ 
improvement in suspicious user detection, and about $50\%$ 
improvement in suspicious message detection, 
which outperforms the baselines.
As Figure\ref{fig:wbeaglemine} shows, the anomalous users detected
by Fraudar and SpokEn only fall in cluster \textcircled{1},
while \eaglemine detects suspicious users from other clusters as well. 
Simply put, the baselines miss most of the anomalies in clusters 
\textcircled{2} and \textcircled{3}, while \eaglemine catches them.

Finally, in Figure~\ref{fig:wbeaglemine},
we also found that the bottom left cluster consists of an isolated network, 
in which users retweeted different messages from those retweeted 
by other users out of network. 
In other words, users in the network did not ``connect'' to 
the outside users via messages in the bipartite graph (user-retweet-message).

\subsection{Q4. Scalability}
Figure~\ref{fig:run_time_fit}
shows the near-linear scaling of \Eaglemine's running time 
in the numbers graph nodes. 
Here we used Sina Weibo dataset, 
and subsampled the nodes to generate different size graphs.

\section{Conclusion}
\label{sec:conclusions}
We propose a tree-based approach \eaglemine to 
mine and summarize all node groups in a 
histogram plot of a large graph.
The \eaglemine algorithm finds optimal clusters 
based on water-level tree and hypothesis test, 
describes them with a configurable model vocabulary, 
and detects some suspicious.
\eaglemine has desirable properties:
\begin{compactitem}
	\item \textbf{Automated summarization:} 
	  Our algorithm automatically summarizes a given histogram 
	  with vocabulary of distributions, inspired by human vision to 
	  find the graph node groups and outliers. 
	\item \textbf{Effectiveness:}
	  We compared \eaglemine on real data with the baselines, the result shown 
	  that our detection is consistent with human vision, achieves better MDL 
	  in summarization, and has better accuracy in anomaly detection.
	\item \textbf{Scalability:}
	 \eaglemine algorithm runs near linear in the number of graph nodes.
\end{compactitem}

%
%
%
\bibliographystyle{IEEEtran}
\balance
\bibliography{EagleMine}

\end{document}